\documentclass[10pt,conference]{IEEEtran}
\IEEEoverridecommandlockouts\IEEEpubid{\makebox[\columnwidth]{978-1-7281-2700-2/19/\$31.00 2019 $\copyright$ IEEE \hfill}\hspace{\columnsep}\makebox[\columnwidth]{ }}
\usepackage{cite}
\usepackage{amsmath,amssymb,amsfonts}
\usepackage{graphicx}
\usepackage{xcolor}
\usepackage{color, colortbl}
\usepackage{amsthm,amssymb}
\usepackage{wasysym}

\usepackage{subcaption}

\usepackage[acronym]{glossaries}
\usepackage{algorithm}
\usepackage[noend]{algpseudocode}
\makeatletter
\renewcommand{\ALG@beginalgorithmic}{\small}
\makeatother
\usepackage{multirow}
\usepackage{booktabs}
\usepackage{graphicx}
\usepackage{fancyvrb}
\algrenewcommand\algorithmicindent{1.4em}

\usepackage{todonotes}

\usepackage{cleveref}
\crefname{section}{\S}{\S\S}
\Crefname{section}{\S}{\S\S}

\usepackage{tikz}
\usetikzlibrary{snakes, matrix, arrows, calc, shadows, automata, shapes.geometric, chains, fit, shapes.multipart,chains,positioning}
\def\firstcircle{(90:0.3cm) circle (0.4cm)}
\def\secondcircle{(210:0.3cm) circle (0.4cm)}
\def\thirdcircle{(330:0.3cm) circle (0.4cm)}
\usetikzlibrary{patterns}
\usetikzlibrary{backgrounds}
\usetikzlibrary{arrows.meta,positioning}

\theoremstyle{remark}
\newtheorem*{theorem}{Theorem}

\newtheorem*{definition}{Definition}
\newtheorem*{example}{Example}

\usepackage[all]{xy}

\newcommand{\disjointrangesfig}[3]{%
\begin{tikzpicture}[->,>=stealth',shorten >=1pt,auto,node distance=0.7cm,semithick, initial text=]%
 \tikzstyle{every state}=[fill=cyan,draw=none,scale=1]%
 \tikzset{%
   simple node/.style={%
     draw,%
     text height=2.2ex,%
     text depth=1ex,%
     inner sep=0pt,text width=5.8ex,align=center%
   },%
   split node/.style={%
     simple node,%
     rectangle split,rectangle split horizontal,%
     draw,inner sep=0ex,rectangle split part align=base,%
   },%
 }%
 \node[split node, rectangle split parts=6] (localbuf) at (0.39, -2.7) {#3};%
 \node (j) at (-2.022+#1*0.54, -3.45) {\texttt{begin}};%
 \draw[->] (j) + (0, 0.2) -> ++ (0, .5);%
 \node (j) at (-0.955+#2*0.645, -3.45) {\texttt{end}};%
 \draw[->] (j) + (0, 0.2) -> ++ (0, .5);%
\end{tikzpicture}
}

\newcommand*\circled[2]{\tikz[baseline=(char.base)]{
            \node[shape=circle,draw,fill=#1,inner sep=1pt] (char) {#2};}}
\newcommand{\ruleone}{\circled{gray!25}{1}}
\newcommand{\ruletwo}{\circled{green!50}{2}}   
\newcommand{\rulethree}{\circled{yellow!50}{3}}

\newif\ifipprefixexeample
\begin{document}
\def\mystrut(#1,#2){\vrule height #1pt depth #2pt width 0pt}

\newacronym[longplural={packet equivalence classes}]{pec}{PEC}{packet equivalence class}
\newcommand{\ddnf}{ddNF}
\newcommand{\dagnode}{n}
\newcommand{\devicebox}{\nu}
\newcommand{\sharppec}{\#PEC}
\newcommand{\diekmann}[1]{{\small\textsc{Diekmann$\setminus\mathsf{#1}$}}}
\newcommand{\diekmannall}{{\small\textsc{Diekmann$\setminus\mathsf{A}, \ldots$}}}
\newcommand{\berkeleyip}{{\small\textsc{Berkeley-IP}}}
\newcommand{\stanfordip}{{\small\textsc{Stanford-IP}}}
\newcommand{\stanfordacl}{{\small\textsc{Stanford-ACL}}}
\newcommand{\pec}[1]{\mathsf{PEC}(#1)}

\title{A Precise and Expressive Lattice-theoretical Framework for Efficient Network Verification}

\author{
{\rm Alex Horn}\IEEEauthorrefmark{1}\thanks{$^{*}$ This work was completed at Fujitsu Laboratories of America.}\\
Apple
\and
{\rm Ali Kheradmand}\IEEEauthorrefmark{1}\\
University of Illinois at Urbana-Champaign
\and
{\rm Mukul R. Prasad}\\
Fujitsu Laboratories of America
% copy the following lines to add more authors
% \and
% {\rm Name}\\
%Name Institution
% end author
%\IEEEcompsocitemizethanks{
%     \IEEEcompsocthanksitem Alex\ Horn 
%     \IEEEcompsocthanksitem Ali\ Kheradmand}
%\thanks{This work was completed at Fujitsu Labs}
}

%\author{Alex Horn\inst{1} \and Ali Kheradmand\inst{2} \and Mukul R. Prasad\inst{1} \and Paparao Palacharla\inst{1} \and Naoki Oguchi\inst{3}}
%\institute{\footnotesize Fujitsu Labs of America and \textsuperscript{3}~Fujitsu Labs Limited \and University of Illinois at Urbana-Champaign}

\maketitle

\begin{abstract}

%Efficient Network Verification Using Lattice Theory and Model Counting
Network verification promises to detect errors, such as black holes and forwarding loops, by logically analyzing the control or data plane. To do so efficiently, the state-of-the-art (e.g., Veriflow) partitions packet headers with identical forwarding behavior into the same \gls*{pec}.

Recently, Yang and Lam showed how to construct the minimal set of \glspl*{pec}, called atomic predicates. Their construction uses Binary Decision Diagrams (BDDs). However, BDDs have been shown to incur significant overhead per packet header bit, performing poorly when analyzing large-scale data centers. The overhead of atomic predicates prompted \ddnf\ to devise a specialized data structure of Ternary Bit Vectors (TBV) instead.

However, TBVs are strictly less expressive than BDDs. Moreover, unlike atomic predicates, \ddnf's set of \glspl*{pec} is not minimal. We show that \ddnf's non-minimality is due to empty \glspl*{pec}. In addition, empty \glspl*{pec} are shown to trigger wrong analysis results. This reveals an inherent  tension between precision, expressiveness and performance in formal network verification.

Our paper resolves this tension through a new lattice-theoretical \gls*{pec}-construction algorithm, \sharppec, that advances the field as follows: (i)~\sharppec\ can encode more kinds of forwarding rules (e.g., ip-tables) than \ddnf~and~Veriflow, (ii)~\sharppec\ verifies a wider class of errors (e.g., shadowed rules) than \ddnf, and (iii) on a broad range of real-world datasets,~\sharppec\ is $10\times$ faster than atomic predicates. By achieving precision, expressiveness and performance, this paper answers a longstanding quest that has spanned three generations of formal network analysis techniques.
\end{abstract}

%\begin{IEEEkeywords}
%Network, Verification, Theory
%\end{IEEEkeywords}

\section{Introduction}
\label{sec:intro}

% Prior PEC-based solutions are either fast or guaranteed to produce correct results, but not both. #PEC is both.

In complex networks, misconfigurations continue to be common~\cite{SHSKRS2012,crystalnet}, causing costly unscheduled outages or compromising security~\cite{W2004,HA2006,W2010}. This has generated significant interest in formally analyzing network behavior on the control (e.g.,~\cite{FFPWGMM2015,GVAM2016,era,BGMW2017}) or data plane (e.g.,~\cite{MKACGK2011,KVM2012,KCZVMcKW2013,Z2014,HKP2017,Backesetal2019}), a class of formal methods collectively known as \emph{network verification}. In this paper, we provide a new algorithm and data structure that can serve as a foundation for both forms of network verification.

%Our lattice-theoretical framework applies to both data plane and control plane checkers.

What make network verification interesting is its predictive power: it promises to find network-related errors that traditional diagnostic tools, such as ping and trace\-route, in general cannot. To accomplish this feat, network verification creates a mathematical model of the network to logically analyze the packet forwarding behavior of packets, rather than merely observing network traffic. This is an inherently difficult task: even reachability checking in the data plane is NP-hard~\cite{MKACGK2011}. Much research therefore has gone into making formal network analysis as efficient as possible.

%
%This 
%
%Our lattice-theoretical framework offers new insights into the NP-hard problems faced by both data plane and control plane checkers must solve.

% Much research therefore has gone into making formal network analysis as efficient as possible. This is an inherently difficult task: even reachability checking in the data plane is NP-hard~\cite{MKACGK2011}. Much research therefore has gone into making formal network analysis as efficient as possible.

\emph{First-generation} formal network analysis tools (e.g.,~\cite{XZMZGHR2005,JS2009,NBDFK2010,AA2010,MKACGK2011,SSYPG2012,ZMMN2012,JBOK2014,BBGIKSSV2014,MCD2015}) rely on \emph{SAT/SMT solvers}, highly optimized backtracking decision procedures for solving propositional or first-order logic problems. However, SAT/SMT solvers are too slow to enumerate all witnesses of each network property violation~\cite{LBGJV2015}, and SAT/SMT solvers tend to perform poorly on reachability queries over many distinct network paths~\cite{Backesetal2019}.

This bottleneck prompted \emph{second-generation} formal network analysis techniques to use a geometric model for packet classification instead, notably in the form of \emph{Header Space Analysis} (HSA)~\cite{KVM2012,KCZVMcKW2013,HHAGJ2014}. At its core, HSA represents packet headers as the difference of cubes in a multi-dimensional hyperspace. While compact, a significant drawback of HSA's difference of hypercube representation is that it is computationally expensive to evaluate in general. This explains why HSA uses a lazy evaluation strategy, which still has performance problems (akin to lazy functional languages).

By contrast, \emph{third-generation} formal network analysis tools avoid the problems of lazy evaluation by \emph{pre-computing} a family of disjoint sets of packet headers. We call these \emph{packet equivalence classes} (\glspl*{pec}). Intuitively, each \gls*{pec} contains packet headers that experience the same forwarding behavior through the network at each router---a form of lossless compression that has been shown to make formal network analysis more efficient in both time and space~\cite{YL2013}.\footnote{While, in the worst case, the number of generated \glspl*{pec} is exponential in the number of match conditions, in practice there are only relatively few \glspl*{pec}~\cite{YL2013,BJMSV2016}. In fact, in restricted, but not uncommon cases, the number of \glspl*{pec} is even linear in the number of match conditions~\cite{McG2012,HKP2017}.}

Formal network analysis tools based on \glspl*{pec} include Veri\-flow~\cite{KZZCG2013}, APV~\cite{YL2013}, \ddnf~\cite{BJMSV2016} and Delta-net~\cite{HKP2017}, all of which detect a myriad of network errors---such as black holes, forwarding loops, reachability and isolation violations---and \glspl*{pec} help to do so in a vendor-agnostic manner. In this paper, we focus on Veriflow~\cite{KZZCG2013}, APV~\cite{YL2013} and \ddnf~\cite{BJMSV2016}. These tools can encode match conditions with possibly many packet header fields, so-called \emph{multi-dimensional match conditions}.

However, reasoning about multi-dimensional match conditions in a priority-ordered list (such as a forwarding table) is challenging, because a higher-priority rule $x$ may overlap with a lower-priority rule $y$. Such overlapping amounts to logical negation (i.e., $y \wedge \neg x$), because $x$ needs to be subtracted from $y$. The crux of the problem is that logical negations can lead to an exponential number of case splits. Consider some packet header filter that uses the match condition $\texttt{1$\ast$1$\ast$0}$, an instance of a \emph{Ternary Bit Vector} (TBV) where `$\ast$' matches either `1' or `0'. The number of case splits due to TBV-negation, such as $\neg(\texttt{1$\ast$1$\ast$0})$, is generally exponential in the length of the TBV.

\emph{Binary Decision Diagrams} (BDDs)~\cite{B1986,K2009} can efficiently represent such case splits, and APV~\cite{YL2013} uses BDDs to compactly represent the space of packet headers, including their negation. By constructing BDDs, APV produces also canonical and optimal \glspl*{pec}, called \emph{atomic predicates}, which form the unique and smallest partition of packet headers~\cite{YL2013}.

But there is a catch: BDDs incur significant overhead per bit in each packet header field, a performance bottleneck in real-world network analysis~\cite{BJMSV2016}. This prompted \ddnf\ to not use BDDs. Instead, \ddnf\ constructs \glspl*{pec} by only intersecting TBVs. The intersection of TBVs is very efficient due to their compact representation in memory, and experiments using Azure data center snapshots confirm that \ddnf\ is significantly more efficient than APV, a remarkable achievement.

However, both \ddnf's TBVs as well as Veriflow's multi-dimensional trie data structure have inherent limitations (\cref{subsec:expressiveness}): they cannot efficiently represent match conditions over arbitrary sets and ranges of ports, and their complements. Consequently, \ddnf\ and Veriflow cannot analyze common firewall rules in practice (\cref{sec:evaluation}), such as iptables rule-sets~\cite{DMHC2016}.

Furthermore, \ddnf\ and Veriflow's \glspl*{pec} are not minimal (\cref{subsec:precision}). In the case of \ddnf, we show that this non-minimality can lead to wrong analysis results, e.g., \ddnf\ is unsuitable for detecting shadowed rules. We catalog over forty cases of such imprecision (\cref{subsec:case-study}). This motivates the following question:

\begin{quote}
\textit{Can the construction of precise and expressive packet equivalences classes be also efficient?}
\end{quote}

Our paper answers this question in the affirmative through a new lattice-theoretical \gls*{pec}-construction algorithm (\cref{sec:algorithms-and-data-structures}), \sharppec, that combines the precision and expressiveness of atomic predicates with the scalability of \ddnf. \sharppec\ is more expressive than Veriflow and \ddnf, because its encoding is not tied to TBVs. As a result, for instance, \sharppec\ can check match conditions with arbitrary ranges, e.g., iptables rule-sets. Moreover, \sharppec\ can detect errors, such as shadowed rules, that \ddnf\ cannot in general, since its analysis is imprecise.

%We achieve this expressiveness through a collection of abstract data types, called element types.

%%%%%%%%%%%%%%%w

We show that \ddnf's imprecision is due to \glspl*{pec} that are empty. We detect such empty \glspl*{pec}---a coNP-hard problem---by efficiently counting the packet headers in each \gls*{pec}. This way \sharppec\ achieves full precision, and it does so $10-100\times$ faster than SAT/SMT and BDD-based solutions that encode the \gls*{pec}-emptiness problem into propositional logic.
%resolving the tension between the precision and performance in negation-free \gls*{pec}-construction.
Moreover, by detecting empty \glspl*{pec}, \sharppec\ constructs \glspl*{pec} that are unique and minimal (\Cref{theorem:optimality}), achieving the optimality of atomic predicates, but at least $10\times$ faster than APV (\cref{sec:evaluation}).

To avoid the aforementioned limitations of TBVs and multi-dimensional trie data structures, we organize packet headers in a \emph{meet-semilattice}~\cite{DP2002} (\cref{subsec:data-structures}). Through this lattice-theoretical framework, \sharppec\ can formally analyze a strictly broader class of forwarding filters than \ddnf~and~Veriflow.

By achieving precision, expressiveness and performance, we answer a longstanding quest that has spanned three generations of formal network analysis techniques.

\section{Background and Motivation}
\label{sec:motivation}

\begin{figure}[t]
\centering
  \begin{tabular}{|c|r|r|r|r|} \hline
      {}                           &
      \multicolumn{1}{c|}{{\small\textsc{Source}}}      &
      \multicolumn{1}{c|}{{\small\textsc{Destination}}} &
      \multicolumn{1}{c|}{{\small\textsc{Proto}}}    &
      \multicolumn{1}{c|}{{\small\textsc{Action}}} \\ \hline

       \ruleone & {\small\texttt{0.0.0.4/30}}    & {\small\texttt{0.0.0.0/28}} & {\small\texttt{!UDP}}  & {\small\texttt{DROP}}         \\
     \ruletwo & {\small\texttt{0.0.0.0/29}}  &  {\small\texttt{0.0.0.12/30}}  &  {\small\texttt{UDP}}  & {\small\texttt{DROP}}           \\ 
            \rulethree & {\small\texttt{0.0.0.4/30}}    & {\small\texttt{0.0.0.12/30}}  & {\small\texttt{ANY}}    & {\small\texttt{FORWARD}}          \\ \hline
  \end{tabular}
  \caption{Forwarding table (using priorities) with 3-dimensional match conditions that neither Veriflow nor \ddnf\ can analyze}\label{fig:motivation-table}
\end{figure}

We start by giving background on formal network analysis (\cref{subsec:background}), illustrating why achieving expressiveness (\cref{subsec:expressiveness}) and precision (\cref{subsec:precision}) at the same time is challenging.

\subsection{Background: Formal Network Analysis}
\label{subsec:background}

In this subsection, we explain through illustrations what makes multi-dimensional match conditions challenging to formally analyze. Readers familiar with \gls*{pec}-based formal network analysis may wish to skip this subsection for now.

Consider two physically connected routers $\devicebox_1$ and $\devicebox_2$. The network operator wants to check the absence of  forwarding loops between $\devicebox_1$ and $\devicebox_2$. 
%For simplicity, assume for now that $\devicebox_1$ forwards all packets to $\devicebox_2$, and $\devicebox_2$ in turn forwards packets back to $\devicebox_1$ according to the forwarding table in~\Cref{fig:motivation-table}. 
Assume that $\devicebox_2$ forwards packets to $\devicebox_1$ according to the forwarding table in~\Cref{fig:motivation-table}.
This forwarding table has three priority-ordered rules: \ruleone, \ruletwo\ and \rulethree, where \ruleone\ has highest priority. Since the match conditions of \ruleone, \ruletwo\ and \rulethree\ filter packets based on three packet header fields, they are instances of 3-dimensional match conditions.

\begin{figure}[b]
  \centering
    \begin{tikzpicture}[>=stealth',
    shorten > = 1pt,
node distance = 0cm and 2.7cm,
    el/.style = {inner sep=2pt, align=left, sloped},
every label/.append style = {font=\tiny}
                    ]
    \node (n0) {\includegraphics[scale=0.6]{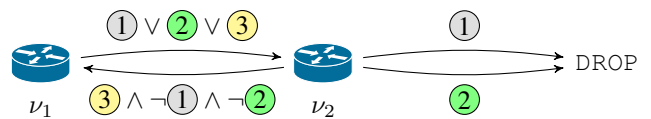}};
    \node[below=of n0] (n0desc) {$\devicebox_1$};
    \node[right=of n0] (n1) {\includegraphics[scale=0.6]{router.eps}};
    \node[right=of n1] (n2) {\texttt{DROP}};
    \node[below=of n1] (n1desc) {$\devicebox_2$};
    \path [->] (n0) edge [bend left=7] node[above] {$\ruleone \vee \ruletwo \vee \rulethree$} (n1)
               (n1) edge [bend left=7] node[below] {$\rulethree \wedge \neg\ruleone \wedge \neg\ruletwo$} (n0)
               (n1) edge [bend left=7] node[above] {\ruleone} (n2)
               (n1) edge [bend right=7] node[below] {\ruletwo} (n2);
  \end{tikzpicture}
    \caption{Verification: is there a forwarding loop, or not?}\label{fig:motivation-graph}
\end{figure}

Consider the lowest-priority rule \rulethree\ in~\Cref{fig:motivation-table}. It is not difficult to see that the set of packets matched by \rulethree\ correspond to the logic formula $\rulethree \wedge \neg\ruleone \wedge \neg\ruletwo$. This formula says that \rulethree\ matches only packet headers that are not matched by \ruleone\ or \ruletwo, thereby encoding the fact that both \ruleone\ and \ruletwo\ have a higher priority than \rulethree.

%To understand the significance of such logic formulas, let us refine our initial assumption about $\devicebox_1$'s forwarding behavior and assume that it forwards to $\devicebox_2$ all packets matched by either \ruleone, \ruletwo\ or \rulethree, i.e., $\ruleone \vee \ruletwo \vee \rulethree$.
To understand the significance of such logic formulas, assume that $\devicebox_1$ forwards to $\devicebox_2$ all packets matched by either \ruleone, \ruletwo\ or \rulethree, i.e., $\ruleone \vee \ruletwo \vee \rulethree$.
Abstractly, formal analysis tools essentially reason about the forwarding behavior of a network in terms of a directed graph whose edges are annotated by such logic formulas (or \glspl*{pec} as we shall see), as illustrated in~\Cref{fig:motivation-graph}. The existence of a forwarding loop between $\devicebox_1$ and $\devicebox_2$ depends on whether the logic formula $\phi = (\ruleone \vee \ruletwo \vee \rulethree) \wedge (\rulethree \wedge \neg\ruleone \wedge \neg\ruletwo)$ is satisfiable or not; equivalently, does there exist a packet header such that  formula $\phi$ can evaluate to true?

The challenge for \gls*{pec}-based formal analysis tools is to be able to express complex multi-dimensional match conditions, while also being able to efficiently and precisely solve the resulting constraint systems via \glspl*{pec}. Unlike SAT/SMT solvers, \glspl*{pec} give by default the set of all such solutions (if any).

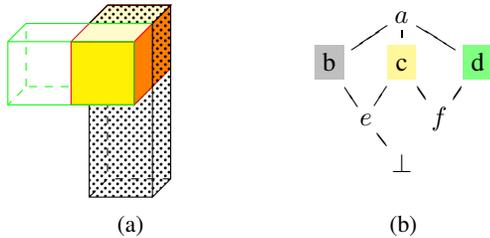
\begin{figure}[t]
%\vspace{-1em}
\centering
  \begin{subfigure}[b]{0.37\columnwidth}
\begin{tikzpicture}[scale=0.21]
\pgfmathsetmacro{\cubex}{4}
\pgfmathsetmacro{\cubey}{4}
\pgfmathsetmacro{\cubez}{6}
      \draw[green,style=solid] (0,0,0) -- ++(0,0,-3) -- ++ (0,-4,0) -- ++ (0,0,3);
      \draw[green,style=dashed] (-5.3,-4,-3) -- ++(-2.7,0,0) -- ++ (0,0,3);
      \draw[green,style=dashed] (-8,-4,-3) -- ++(0,4,0);
      \draw[black,style=solid,fill=gray!25] (0,0,-3) -- ++(-4,0,0) -- ++(0,0,-3) -- ++(4,0,0) -- cycle;
     \draw[black,style=dashed] (-4,0,-6) -- ++(0,-11,0) -- ++(0,0,3);
      \draw[black,style=dashed] (0,-11,-6) -- ++(-4,0,0);
%      \draw[red,style=solid] (0,0,0) -- ++(-4,0,0) -- ++(0,-4,0) -- ++(4,0,0) -- cycle;
%      \draw[red,style=solid] (0,0,0) -- ++(0,0,-6) -- ++(0,-4,0) -- ++(0,0,6);
%      \draw[red,style=solid] (0,0,-6) -- ++(-4,0,0) -- ++(0,0,6);
      %\draw[red,style=dotted] (-4,-4,0) -- ++(0,0,-6) -- ++(4,0,0);
      %\draw[red,style=dotted] (-4,-4,-6) -- ++(0,4,0);

      %\draw[black,style=dashed] (0,0,0) -- ++(0,0,-5) -- ++(-7,0,0) -- ++ (0,-3,0) -- ++ (7,0,0) -- cycle;
      \draw[black,style=solid,fill=gray!15,pattern=crosshatch dots,pattern] (0,-11,-3) -- ++(-4,0,0) -- ++(0,5.8,0) -- ++(4,2,0);
      \draw[red,fill=yellow] (0,0,0) -- ++(-\cubex,0,0) -- ++(0,-\cubey,0) -- ++(\cubex,0,0) -- cycle;
      \draw[red,fill=orange] (0,0,0) -- ++(0,0,-\cubez) -- ++(0,-\cubey,0) -- ++(0,0,\cubez) -- cycle;
      \draw[red,fill=yellow!25] (0,0,0) -- ++(-\cubex,0,0) -- ++(0,0,-\cubez) -- ++(\cubex,0,0) -- cycle;
      \draw[solid,style=dotted] (0,-\cubey,-3) -- ++(0,\cubey,0) -- ++(-\cubex,0,0);
      \draw[solid,style=dotted,style=solid,pattern=crosshatch dots,pattern] (0,0,-3) -- ++(-\cubex,0,0) -- ++(0,0,-3) -- ++ (\cubex,0,0) -- cycle;
      \draw[black,style=solid,pattern=crosshatch dots,pattern] (0,0,-3) -- ++(0,-11,0) -- ++(0,0,-3) -- ++(0,11,0);

      \draw[green,style=solid] (0,0,0) -- ++(-8,0,0) -- ++ (0,-4,0) -- ++ (8,0,0) -- cycle;
            \draw[green,style=solid] (0,0,-3) -- ++(-8,0,0) -- ++ (0,0,3);
\end{tikzpicture}
  \caption{}\label{fig:geometry}
  \end{subfigure}%
\quad%
  \begin{subfigure}[b]{0.37\columnwidth}
    \[
      \xymatrix@C=0em@R=0.37em{
                 &              &a\ar@{-}[lld]\ar@{-}[d]\ar@{-}[rrd]    &             \\
     \colorbox{gray!50}{\mbox{b}}\ar@{-}[rd]&              &\colorbox{yellow!50}{\mbox{\vphantom{b}c}}\ar@{-}[ld]\ar@{-}[rd] &              &\colorbox{green!50}{\mbox{\vphantom{b}d}}\ar@{-}[ld] \\
                 &e\ar@{-}[rd]  &                        & f\ar@{-}[ld] &             \\
                 &              &\bot                    &              &
      }
  \]
  \caption{}\label{fig:hasse-diagram}
  \end{subfigure}%
        \caption{(\subref{fig:geometry}) Geometric view of the three 3-dimensional match conditions in~\Cref{fig:motivation-table}; (\subref{fig:hasse-diagram}) Hasse diagram of the meet-semilattice induced by these match conditions (see also \cref{sec:algorithms-and-data-structures})}
        \label{fig:explicit-pec}
\end{figure}

What does a solution to this challenge entail? To answer this, consider~\Cref{fig:geometry}, a geometric view of the three match conditions in~\Cref{fig:motivation-table}. Since there are 3-dimensional match conditions,~\Cref{fig:geometry} has three axes: the $x$-axis and $y$-axis correspond to the range of the source and destination IP addresses, respectively, whereas the $z$-axis evenly divides the space into \texttt{UDP} and non-\texttt{UDP} packets. The color of each rectangular cuboid corresponds to \ruleone, \ruletwo\ and \rulethree. The key idea behind \glspl*{pec} is to divide the whole geometric space into disjoint sub-spaces prior to the analysis.

Note that each overlapping of cuboids corresponds to an overlapping of a pair of rules. In general, however, reasoning about the intersection of higher-dimensional cuboids, as in~\Cref{fig:geometry}, is NP-hard. For example, there is no forwarding loop between $\devicebox_1$ and $\devicebox_2$, since the 3-dimensional space denoted by $\phi$ is in fact empty, an instance of an NP-hard query.

\subsection{Challenge: Expressiveness}
\label{subsec:expressiveness}

 Notice that the highest-priority rule \ruleone\ in~\Cref{fig:motivation-table} complements an individual packet header field. That is, \ruleone\ matches only \emph{non}-\texttt{UDP} packet headers whose source and destination IP address match {\texttt{0.0.0.4/30}} and {\texttt{0.0.0.0/28}}, respectively. However, the \gls*{pec}-construction schemes in Veriflow and \ddnf\ are not designed for multi-dimensional match conditions with arbitrary ranges, sets of values, or their complements (all of which can be found in iptables rule-sets~\cite{DMHC2016}).
 
Veriflow and \ddnf's limitation is due to the fact that they are tied to TBVs, where Veriflow represents TBVs as a trie data structure with nodes that can have three children for `0', `1' and `$\ast$'. The problem is that a single TBV cannot  represent match conditions such as the non-\texttt{UDP} example above. As another instance, an arbitrary range that is not an IP prefix can only be represented by multiple TBV. This renders the TBV representation of match conditions inefficient and impractical. By contrast, \sharppec\ can efficiently encode such match conditions via element types (\cref{subsec:element-types}). While APV can represent the same match conditions as \sharppec, APV's reliance on BDDs makes it at least $10\times$ slower than \sharppec\ (\Cref{subsec:performance}).

% By contrast, \sharppec\ introduces the notion of element types  (\cref{subsec:element-types}) that can be orders of magnitude more compact (e.g., $32$ bits for a range of 16-bit port numbers).

\subsection{Challenge: Precision and Minimality of PECs}
\label{subsec:precision}

For \ddnf\ to be able to analyze the network in~\Cref{fig:motivation-graph}, let us further simplify the example by replacing the the three match conditions of the rules \ruleone, \ruletwo\ and \rulethree\ with the following three IP prefixes, respectively: $x = \texttt{10.57.0.0/19}$, $y = \texttt{10.57.32.0/19}$ and $z = \texttt{10.57.0.0/18}$. This simplifies the forwarding table in~\Cref{fig:motivation-table} accordingly, where each rule now only matches packets based on longest IP prefix matching---something that \ddnf\ is designed to handle. We remark that our simplification preserves a vital characteristic of the example: reasoning about it requires only \emph{two} \glspl*{pec}, which form atomic predicates by definition (\cref{subsec:minimality}).

However, \ddnf\ constructs \emph{three} \glspl*{pec}, denoted by uppercase letters: $X$ and $Y$ that represent all IP addresses in $x = \texttt{10.57.0.0/19}$, $y = \texttt{10.57.32.0/19}$, respectively, and $Z$ for all IP addresses in $z = \texttt{10.57.0.0/18}$, \emph{except} those IP addresses in $x$ and $y$. By construction, $X$, $Y$ and $Z$ are disjoint, so $\{X, Y, Z\}$ is indeed a set of \glspl*{pec}.

\begin{figure}[t]
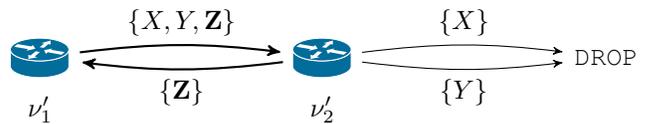

\centering
\begin{tikzpicture}[>=stealth',
    shorten > = 1pt,
node distance = 0cm and 2.7cm,
    el/.style = {inner sep=2pt, align=left, sloped},
every label/.append style = {font=\tiny}
                    ]
    \node (n0) {\includegraphics[scale=0.6]{router.eps}};
    \node[below=of n0] (n0desc) {$\devicebox'_1$};
    \node[right=of n0] (n1) {\includegraphics[scale=0.6]{router.eps}};
    \node[right=of n1] (n2) {\texttt{DROP}};
    \node[below=of n1] (n1desc) {$\devicebox'_2$};
    \path [->] (n0) edge [thick,bend left=7] node[above] {$\{X,Y,\mathbf{Z}\}$} (n1)
               (n1) edge [thick,bend left=7] node[below] {$\{\mathbf{Z}\}$} (n0)
               (n1) edge [bend left=7] node[above] {$\{X\}$} (n2)
               (n1) edge [bend right=7] node[below] {$\{Y\}$} (n2);
  \end{tikzpicture}
    \caption{Wrong result in~\ddnf, due to non-minimal \glspl*{pec}}\label{fig:false-alarm-b}
\end{figure}

The crux of the problem is that $\{X, Y, Z\}$ is not minimal, because there is a \gls*{pec} that is empty, namely $Z$. To illustrate the impact of this superfluous \gls*{pec}, consider~\Cref{fig:false-alarm-b}. Notice that each edge in~\Cref{fig:false-alarm-b} has a corresponding edge in~\Cref{fig:motivation-graph}. The problem is that \ddnf's \gls*{pec} construction fails to precisely capture the Boolean formulas along the edges in~\Cref{fig:motivation-graph}: \ddnf\ wrongly reports a forwarding loop between $\devicebox'_1$ and $\devicebox'_2$, because the edges in~\Cref{fig:false-alarm-b} labelled by $Z$ (shown in bold) form a spurious cycle. This cycle is spurious, and therefore leads to a \emph{false alarm}, because the conjunction of the corresponding Boolean formulas in~\Cref{fig:motivation-graph} is unsatisfiable.\footnote{A Boolean formula is \emph{unsatisfiable} whenever it can never evaluate to true.}

In addition to false alarms, \ddnf's limitation can also manifest itself as a \emph{failure to detect network-related errors}: in our example, \ddnf\ will not detect that the last rule in~\Cref{fig:motivation-graph} is shadowed.\footnote{There are two kinds of shadowed rules: (i)~a single higher-priority rule covers some lower-priority rule, or (ii)~the \emph{union} of several higher-priority rules covers some lower-priority rule. Here \ddnf\ fails to detect the latter.} The possibility for both \emph{false alarms} as well as \emph{false negatives} means that \ddnf\ comes with the overhead to always sanity check its results, a serious limitation (\Cref{subsec:case-study}).

Our experiments (\cref{sec:evaluation}) show that \sharppec\ is $10-100\times$ faster than alternative SAT/SMT solvers and BDD-based solutions for detecting empty \glspl*{pec}. To achieve this speedup, \sharppec\ exploits the fact that it is enough to find the number of packets in a \gls*{pec} to check the \gls*{pec}'s emptiness, rather than finding a witness for its non-emptiness. This reveals, in particular, that in the context of formal network analysis \textit{counting} is much faster than backtracking on a wide range of realistic datasets.

\section{Lattice-theoretical Framework}
\label{sec:algorithms-and-data-structures}

In this section, we highlight the technical approach behind \sharppec\ (\cref{subsec:design}), before explaining its data structures (\cref{subsec:data-structures}--\ref{subsec:pecs-as-dags}) and \gls*{pec}-construction algorithm (\cref{subsec:algorithms}). We explain how to answer queries in \sharppec\ (\cref{subsec:query}). Finally, we show that \sharppec\ constructs the minimal set of \glspl*{pec} (\cref{subsec:minimality}).

% Despite an expected high theoretical worst-case complexity (Appendix~\ref{appendix:complexity}), our experiments (\cref{sec:experiments}) establish that \sharppec\ has clear real-world applications.

\subsection{Technical Approach}
\label{subsec:design}

%To interpret~\Cref{fig:simple-hasse-diagram},  Following the described DAG construction, $y$ and $z$ are direct children of $x$, because $y \subset x$ and $z \subset x$. There is no edge between $y$ and $z$, because neither is a subset of the other. Furthermore, since the DAG does not include transitive edges, there is no edge from, say, $x$ to the $\bot$.

To illustrate our approach, reconsider the match conditions \ruleone, \ruletwo\ and \rulethree\ in~\Cref{fig:motivation-table}. Recall that \ruleone\ complements an individual packet header field. We can represent such and other match conditions by so-called \emph{element types} (\Cref{fig:element-type}).

The geometric interpretation (\cref{subsec:background}) we considered  in~\Cref{fig:motivation-graph} was only in three dimensions. In general, the geometric view is unfeasible, since it requires reasoning about hypercubes as the number of packet header fields increases. Instead, \sharppec\ constructs a \emph{meet-semilattice}, a form of partially ordered set in which every finite subset of elements has a greatest lower bound~\cite{DP2002}. In doing so, \sharppec\ is able to represent match conditions that \ddnf\ and Veriflow cannot, while also achieving precision and performance, as described below.

%denoted by \colorbox{gray!50}{\mbox{$b$}},\colorbox{yellow!50}{\mbox{\vphantom{$b$}$c$}} and \colorbox{green!50}{\mbox{\vphantom{$b$}$d$}} in~\Cref{fig:lattice-elements}. Note that match condition \colorbox{gray!50}{\mbox{$b$}} complements an individual packet header field. More concretely, \colorbox{gray!50}{\mbox{$b$}} matches only \emph{non}-\texttt{UDP} packet headers whose source and destination IP address match {\texttt{0.0.0.4/30}} and {\texttt{0.0.0.0/28}}, respectively. Such match conditions cannot be generally encoded using~\ddnf's TBVs. Instead, 

%Since each match condition filters packets based on three packet header fields,~\Cref{fig:explicit-pec} gives a three-dimensional geometric interpretation where the $x$-axis and $y$-axis correspond to the range of the source and destination IP addresses, respectively, whereas the $z$-axis evenly divides the space into \texttt{UDP} and non-\texttt{UDP} packets. The color of each rectangular cuboid corresponds to \colorbox{gray!50}{\mbox{$b$}},\colorbox{yellow!50}{\mbox{\vphantom{$b$}$c$}} and \colorbox{green!50}{\mbox{\vphantom{$b$}$d$}} in~\Cref{fig:lattice-elements}, respectively.

\begin{figure}[b]
\centering
  \begin{tabular}{|c|rrr|} \hline
      {}                           &
      \multicolumn{1}{c|}{{\textsc{Source}}}      &
      \multicolumn{1}{c|}{{\textsc{Destination}}} &
      \multicolumn{1}{c|}{{\textsc{Protocol}}}    \\ \hline
       $a$ & {\texttt{0.0.0.0/0}}    & {\texttt{0.0.0.0/0}}   & {\texttt{ANY}}            \\
       \cellcolor{gray!25}$b$ & \cellcolor{gray!25}{\texttt{0.0.0.4/30}}    & \cellcolor{gray!25}{\texttt{0.0.0.0/28}} & \cellcolor{gray!25}{\texttt{!UDP}}            \\
      \cellcolor{yellow!50}$c$ & \cellcolor{yellow!50}{\texttt{0.0.0.4/30}}    & \cellcolor{yellow!50}{\texttt{0.0.0.12/30}}  & \cellcolor{yellow!50}{\texttt{ANY}}            \\
      \cellcolor{green!50}$d$ & \cellcolor{green!50} {\texttt{0.0.0.0/29}}  & \cellcolor{green!50} {\texttt{0.0.0.12/30}}  & \cellcolor{green!50} {\texttt{UDP}}            \\
       $e$ & {\texttt{0.0.0.4/30}}  & {\texttt{0.0.0.12/30}}  & {\texttt{!UDP}} \\
       $f$ & {\texttt{0.0.0.4/30}}  & {\texttt{0.0.0.12/30}}  & {\texttt{UDP}}    \\ \hline
  \end{tabular}
  \caption{Six semi-meetlattice elements induced by the three 3-dimensional match conditions \colorbox{gray!50}{\mystrut(6.9,1) $b$},\colorbox{yellow!50}{\mystrut(6.9,1) $c$} and \colorbox{green!50}{\mystrut(6.9,1) $d$} where \colorbox{gray!50}{\mystrut(6.9,1) $b$} features negation of a protocol header field, e.g., `\texttt{!UDP}'}\label{fig:lattice-elements}
  %\caption{Match conditions () and lattice elements}\label{fig:lattice-elements}
\end{figure}

\Cref{fig:hasse-diagram} shows the Hasse diagram of the meet-semilattice produced by \sharppec, given the match conditions in~\Cref{fig:motivation-table}. A \emph{Hasse diagram} has an edge from a vertex $v$ to a vertex $u$ whenever $u$ is a subset of $v$ (written $u \subset v$), and there is no other vertex $w$ such that $u \subset w \subset v$. In other words, only non-transitive edges are included in the Hasse diagram.

\Cref{fig:lattice-elements} gives the more familiar interpretation of the elements in the meet-semilattice as match conditions. Note that the color of the rows in \Cref{fig:lattice-elements} corresponds to the coloring of the corresponding match conditions in~\Cref{fig:motivation-table}.

Observe that the meet-semilattice in~\Cref{fig:hasse-diagram} contains more elements than there are match conditions. Intuitively, the reason is that the meet-semilattice describes the overlapping of all match conditions. This intuition is made precise by the requirement that every meet-semilattice is closed under intersection: it must contain every element that is the result of intersecting sets of other elements. For example, elements $e$ and $f$ are in the meet-semilattice because $e = \colorbox{gray!50}{\mbox{b}} \cap \colorbox{yellow!50}{\mbox{\vphantom{b}c}}$ and $f= \colorbox{yellow!50}{\mbox{\vphantom{b}c}} \cap \colorbox{green!50}{\mbox{\vphantom{b}d}}$, respectively. The last two rows in \Cref{fig:lattice-elements} give a more familiar interpretation of elements $e$ and $f$.

\begin{figure}[t]
    \begin{align*}
      A&\triangleq a - (\colorbox{gray!50}{\mbox{b}} \cup \colorbox{yellow!50}{\mbox{\vphantom{b}c}} \cup \colorbox{green!50}{\mbox{\vphantom{b}d}}) &  C&\triangleq \colorbox{yellow!50}{\mbox{\vphantom{b}c}} - (e \cup f)  & E&\triangleq e    \\
      B&\triangleq \colorbox{gray!50}{\mbox{b}} - e     &   D&\triangleq \colorbox{green!50}{\mbox{\vphantom{b}d}} - f    & F&\triangleq f
      \end{align*}
  \caption{\glspl*{pec} due to the match conditions (colored rows) in~\Cref{fig:lattice-elements}, using the Hasse diagram in~\Cref{fig:hasse-diagram}}\label{fig:pecs-def}
\end{figure}

\sharppec\ bases its meet-semilattice construction on the algorithm in~\cite{KOWvdM2009}, but with a twist: we maintain the \emph{cardinality} of each \gls*{pec}---the number of packet headers in each \gls*{pec}. Crucially, an empty \gls*{pec} has cardinality zero. This way, \sharppec\ detects that $e \cup f = \colorbox{yellow!50}{\mbox{\vphantom{b}c}}$, which \ddnf\ cannot. Unlike a \emph{per bit} combinatorial backtracking search with SAT/SMT solvers, our cardinality computation uses the structure of the semilattice and harnesses the computing power of ALUs~\cite{AgnerFog} (\ref{subsec:pec-emptiness-checking-evaluation}). 

Next, we discuss the technical details behind \sharppec, specifically: its data structures for representing match conditions (\cref{subsec:element-types}) and \glspl*{pec} (\cref{subsec:pecs-as-dags}), as well as its algorithm that use these data structures to compute \glspl*{pec} (\cref{subsec:algorithms}) and answer a network operator's queries about the network (\cref{subsec:query}).

%MENTION:
%\sharppec\ is designed to be extensible so it can analyze a diverse class of networks.
%\footnote{Since \sharppec's network analysis (which is not discussed in detail here) is carried out with respect to \glspl*{pec}, it is independent of element types, which ensures that policy checking algorithms as in~\cref{subsec:pecs} are standalone and reusable across different networks.} Similar to \ddnf~\cite{BJMSV2016}, \sharppec's representation of \glspl*{pec} as a Hasse diagram (e.g., \Cref{fig:hasse-diagram}) via a DAG avoids the need for BDDs,
%because no complements of hypercubes need to be computed.
%(complement-free) DAG-based representation

%\subsection{Data Structures}
%\label{subsec:data-structures}

\subsection{Representation of Match Conditions via Element Types}
\label{subsec:data-structures}
\label{subsec:element-types}

\begin{figure}[b]
\centering
  \begin{tabular}{|l|l|} \hline
      \multicolumn{1}{|c|}{\textsc{Element Type}} &
      \multicolumn{1}{c|}{{\textsc{Description}}}             \\ \hline
      {\texttt{ip\_prefix}} & IP prefix, convertible to \texttt{range} \\
      {\texttt{optional<T>}} & Wildcard or a value of type {\texttt{T}} \\
      {\texttt{tbv<N>}} & Fixed-length TBV \\
      {\texttt{range}} & Half-closed interval, e.g., $[0:10)$ \\
      {\texttt{disjoint\_ranges}} & Set of disjoint ranges \\
      {\texttt{set<T>}} & Finite set of values of type {\texttt{T}} \\
      {\texttt{tuple<E$_1, \ldots, $E$_k$>}} & Tuple where {\texttt{E}$_j$} are element types \\ \hline
  \end{tabular}
 \caption{Element types to form complex (i.e., multi-dimensional) packet header match conditions}\label{fig:element-type}
\end{figure}

At its core, \sharppec\ features an abstract data type for match conditions, called \emph{element type}, which strictly generalizes the expressiveness of Veriflow and \ddnf's TBVs. For example, using element types, we encode the match conditions in~\Cref{fig:lattice-elements} as 3-tuples $\langle F_1, F_2, F_3 \rangle$ where $F_1$ and $F_2$ denote ranges of source and destination IP addresses, respectively, whereas $F_3$ denotes a set of protocols where `!' on the protocol field is encoded by efficiently complementing a bitset.

For its generalization, \sharppec\ imposes only two basic requirements on element types: elements must form a finite partially ordered set, whose cardinality must be computable in polynomial time. \Cref{fig:element-type} shows fundamental element types used in practice that satisfy these requirements where the partial ordering corresponds to the usual subset inclusion order. For example, if $x$ and $y$ are of type \texttt{ip\_prefix}, $x \subseteq y$ means that every IP address in $x$ appears also in $y$.

Each element type features three operators: equality ($=$), intersection ($\cap$), and $\mathsf{cardinality}$. All element types in~\Cref{fig:element-type} can be efficiently implemented using data structures that use contiguous memory, and our implementations have therefore high cache locality, similar to TBVs in \ddnf. We remark that since $x \subseteq y$ holds exactly if $x \cap y = x$, the subset-inclusion operator ($\subseteq$) is derived automatically, a default implementation that can be optionally optimized.

Some element types such as {\texttt{disjoint\_ranges}} and {\texttt{set<T>}} where \texttt{T} is a fixed-size type, support a complement (`!') operator. This allows for more complex match conditions, such as complements on protocol fields as in~\Cref{fig:lattice-elements}. By contrast, the \texttt{tuple} element type has no complement operator because it is computationally too expensive~\cite{KVM2012}.

More generally, by introducing element types, \sharppec\ can tightly control the use and effects of complements, allowing only forms of negation that can be efficiently implemented.

\begin{example}
The match conditions in \Cref{fig:motivation-table}, and the corresponding meet-semilattice elements in~\Cref{fig:lattice-elements}, can be represented by  3-tuples of element type {\texttt{tuple<ip\_prefix, ip\_prefix, set<protocol>>}} where {\texttt{protocol}} is an enumeration type. Alternatively, if there is no need to be able to complement the protocol header field, the last tuple component could be also replaced by {\texttt{optional<protocol>}}.
\end{example}

%An important characteristic of the element types in~\Cref{fig:data-types}, particularly {\texttt{tuple<E$_1$, $\ldots$, E$_k$>}} which form partially ordered sets by standard point-wise extension~\cite{DP2002}, is that they allow \sharppec\ to represent match conditions, such as iptables rule-sets in our experiments (\cref{sec:experiments}), that \ddnf\ cannot handle. To illustrate this, observe that element type {\texttt{disjoint\_ranges}} induces a partial order over \emph{sets} of non-overlapping half-closed intervals (i.e., \texttt{range}s). This partial ordering is useful for, say, negating IP prefixes since their negation results in up to two disjoint IP address ranges, including ranges, such as $[0:10)$, that do not precisely correspond to a single IP prefix or standard ternary encoding.

\subsection{\gls*{pec}-representation as a DAG}
\label{subsec:pecs-as-dags}

\sharppec\ represents the Hasse diagram of a meet-semilattice as a \emph{directed acylic graph} (DAG). Since such a Hasse diagram can be shown to be unique up to graph isomorphism~\cite{DP2002}, so is the DAG that \sharppec\ constructs using the later algorithm.

Therefore, \sharppec\ represents each \gls*{pec} by a pointer to a DAG node. Each such \gls*{pec} denotes the packet headers that are in the element associated with the pointed to DAG node, minus the elements in its children. For example, given the meet-semilattice in~\Cref{fig:hasse-diagram}, uppercase letter $A$ denotes the \gls*{pec} that includes the packet headers in $a$, excluding those in \colorbox{gray!50}{\mbox{$b$}},\colorbox{yellow!50}{\mbox{\vphantom{$b$}$c$}} and \colorbox{green!50}{\mbox{\vphantom{$b$}$d$}}. \Cref{fig:pecs-def} defines the other \glspl*{pec} similarly. By construction, all \glspl*{pec} are pair-wise disjoint. For example, the intersection of the \glspl*{pec} $B$ and $C$ in~\Cref{fig:hasse-diagram} is empty, whereas $\colorbox{gray!50}{\mbox{b}} \cap \colorbox{yellow!50}{\mbox{\vphantom{b}c}}$ is non-empty. Each node $n$ in the DAG has the following three fields: (i)~$n.\mathit{elem}$ denotes the match condition associated with $n$; (ii)~$n.\mathit{children}$ contains all the DAG nodes $c$ such that $c \not= n$ and $c.\mathit{elem} \subseteq n.\mathit{elem}$ and there is no other DAG node $c'$ such that $n \not= c' \not= c$ and $c.\mathit{elem} \subseteq c'.\mathit{elem} \subseteq n.\mathit{elem}$; (iii)~$n.\mathit{cardinality}$ corresponds to the number of packet headers in the \gls*{pec} denoted by $n$.

\begin{example}
Let $n_a$ be the root node of the DAG in~\Cref{fig:hasse-diagram} such that $n_a.\mathit{elem} = a$ and $n_a.\mathit{children}= \{n_b, n_c, n_d \}$. Note that neither $n_f$ and $n_e$ are in $n_a.\mathit{children}$, because they are not direct children of $n_a$. We shall see that \sharppec\ computes $n_c.\mathit{cardinality} = 0$, i.e., the \gls*{pec} denoted by $n_c$ is empty.
\end{example}

% Notice that this space can only be described by the union of at least five rectangular cuboids, which is less compact than \sharppec's representation as a single node. In fact, an explicit hypercube representation is generally not unique, since the space can be partitioned in many different ways. By contrast, \sharppec's representation of \glspl*{pec} is unique, while also avoiding the (generally expensive) computation of hypercubes.

\subsection{Algorithm for Computing \glspl*{pec}}
\label{subsec:algorithms}

The algorithm of \sharppec\ is divided into three procedures, each of which accesses the global variable $\mathit{Modified\_Nodes}$ that determines what \gls*{pec}-cardinalities need to be re-computed. We explain each procedure in turn.

%, and is key for achieving incrementality (\cref{subsec:incrementality}).

\begin{algorithm}[b]
  \caption{Insert new element into meet-semilattice}\label{alg:insert}
  \begin{algorithmic}[1]
%  \Require $\mathit{Root}$ and $\mathit{Modified\_Nodes}$ are global variables
  \Procedure{Insert}{$\mathit{elem}$}
    %\State $\mathit{Modified\_Nodes} \gets \{\}$\label{line:init-modified-nodes}
    \State $\dagnode, \mathit{new} \gets \Call{Find\_Or\_Create\_Node}{\mathit{elem}}$
    \If{$\mathit{new}$}
      \State $\mathit{Modified\_Nodes}.\mathsf{insert}(\dagnode)$
      \State \Call{Insert\_Node}{$\mathit{Root}, \dagnode$} \Comment{$\mathit{Root}.\mathit{elem} = \top$}
      %\State $\mathit{modified\_nodes}' \gets \mathit{Modified\_Nodes}$ \Comment{local copy}\label{line:local-copy-modified-nodes}
      %\For{$\dagnode' \in \mathit{modified\_nodes}'$}\label{line:for-loop-modified-nodes-begin}
      \For{$\dagnode' \in \mathit{Modified\_Nodes}$}\label{line:for-loop-modified-nodes-begin}
        \State \Call{Compute\_Cardinality}{$\dagnode'$}
      \EndFor\label{line:for-loop-modified-nodes-end}
    \EndIf
  \EndProcedure
  \end{algorithmic}
\end{algorithm}

The main procedure, \textsc{Insert} (\Cref{alg:insert}), takes as input an $\mathit{element}$---a match condition of the kind as explained in~\cref{subsec:data-structures}---that is to be added into the meet-semilattice. To do so, \textsc{Insert} calls $\Call{Find\_Or\_Create\_Node}{\mathit{element}}$ which uses a hash table (not shown) to determine when a new DAG node $\dagnode$, satisfying  $\dagnode.\mathit{elem} = \mathit{element}$, has to be created or not. Only in the former case, when $\mathit{new} = \mathbf{true}$, is $\dagnode$ inserted into $\mathit{Modified\_Nodes}$ and the subprocedures \textsc{Insert\_Node} and \textsc{Compute\_Cardinality} are called, as discussed next.

\begin{algorithm}
  \caption{Update DAG representing meet-semilattice}\label{alg:insert-node}
  \begin{algorithmic}[1]
 % \Require $\mathit{Modified\_Nodes}$ is a global variable
  \Procedure{Insert\_Node}{$\mathit{parent}, \dagnode$}
    \State $\Gamma \gets \{\}$
    \For{$\mathit{child} \in \mathit{parent}.\mathit{children}$}
      \If{$\mathit{child}.\mathit{elem} \subseteq \dagnode.\mathit{elem}$}\label{line:cond-0}
        \State $\Gamma.\mathsf{insert}(\mathit{child})$
      \ElsIf{$\dagnode.\mathit{elem} \subseteq \mathit{child}.\mathit{elem}$}\label{line:cond-1}
        \State \Call{Insert\_Node}{$\mathit{child}$, $\dagnode$}
        \State \Return
      \Else\label{line:cond-2}
        \State $e' \gets \dagnode.\mathit{elem} \cap \mathit{child}.\mathit{elem}$
        \If{$e'$ is \textbf{not} empty}\label{line:intersection-begin}
          \State $\dagnode', \mathit{new} \gets \Call{Find\_Or\_Create\_Node}{\mathit{e'}}$
          \State $\Gamma.\mathsf{insert}(\dagnode')$
          \If{$\mathit{new}$}
              \State $\mathit{Modified\_Nodes}.\mathsf{insert}(\dagnode')$\label{line:modified-nodes-update-1}
              \State \Call{Insert\_Node}{$\mathit{child}$, $\dagnode'$}
          \EndIf
        \EndIf\label{line:intersection-end}
      \EndIf
    \EndFor
    \State $\mathit{parent}.\mathit{children}.\mathsf{insert}(\dagnode)$\label{line:update-0}
    \State $\mathit{Modified\_Nodes}.\mathsf{insert}(\mathit{parent})$\label{line:modified-nodes-update-0}
    \State $\mathit{max\_children} \gets $
    \State $\qquad\left\{ c \in \Gamma \mid \forall c' \in \Gamma \colon (c.\mathit{elem} \subseteq c'.\mathit{elem} \rightarrow c = c') \right\}$\label{line:max-children-init}
    \For{$\mathit{max\_child} \in \mathit{max\_children}$}\label{line:max-children-begin}
      \State $\mathit{parent}.\mathit{children}.\mathsf{erase}(\mathit{max\_child})$\label{line:update-1}
      \State $n.\mathit{children}.\mathsf{insert}(\mathit{max\_child})$\label{line:update-2}
    \EndFor\label{line:max-children-end}
  \EndProcedure
  \end{algorithmic}
\end{algorithm}

To insert a new node $n$ into the DAG that represents the Hasse diagram of the meet-semilattice, \textsc{Insert\_Node} is called on the root of the DAG. Intuitively, \textsc{Insert\_Node} (\Cref{alg:insert-node}) works by case analysis on the three possible partial orderings between a pair of nodes (\cref{line:cond-0},~\ref{line:cond-1}~and~\ref{line:cond-2}). By using $\mathit{max\_children}$, we only add a child $c$ to a parent provided that $c$'s element is maximal with respect to the other children elements in $\Gamma$ (\cref{line:max-children-init} and~\ref{line:max-children-begin}--\ref{line:max-children-end}), thereby ensuring that all children of a parent remain mutually incomparable.

The correctness of the induced updates to the DAG edges (\cref{line:update-0},~\ref{line:update-1}~and~\ref{line:update-2}) follows directly from the proof in~\cite{KOWvdM2009}, since we only augment the algorithm by maintaining what nodes have been modified along the way (see $\mathit{Modified\_Nodes}$ on \cref{line:modified-nodes-update-0}~and~\ref{line:modified-nodes-update-1}). There may be multiple such nodes when the insertion of a single new match condition requires multiple DAG nodes to be created, due to the requirement that the lattice be closed under meets, as illustrated next.

\begin{figure}[b]
\centering
\begin{subfigure}[b]{0.1\columnwidth}
  \centering
    \[
      \xymatrix@C=0.2em@R=1em{
          a\ar@{-}[d]& \\
          b\ar@{-}[d]\ar@{-}[rd]& \\
          c\ar@{-}[d]&d\ar@{-}[ldd] \\
          e\ar@{-}[d]& \\
          \bot
      }
  \]
  \caption{}
  \label{fig:second-example-a}
  \end{subfigure}
  \hfill
\begin{subfigure}[b]{0.2\columnwidth}
  \centering
    \[
      \xymatrix@C=0.2em@R=1em{
                                             &a\ar@{-}[rd]\ar@{-}[ld] &                       \\
          b\ar@{-}[d]\ar@{-}[rd]\ar@{-}[drr] &                        &\mathbf{f}\ar@{-}[d]            \\
          c\ar@{-}[rd]                       &d\ar@{-}[rd]            &\mathbf{g}\ar@{-}[d]\ar@{-}[ld] \\
                                             &e\ar@{-}[d]             &\mathbf{h}\ar@{-}[ld]           \\
                                             &\bot                    &
      }
  \]
  \caption{}
  \label{fig:second-example-b}
  \end{subfigure}
  \hfill
  \begin{subfigure}[b]{0.6\columnwidth}
  \centering
  \begin{tabular}{l|l|l}
         & {\small\textsc{Destination}}    & {\small\textsc{Proto}} \\ \hline
     $a$ & {\small\texttt{0.0.0.0/0}}      & {\small\texttt{ANY}}  \\
     $b$ & {\small\texttt{210.4.214.0/23}} & {\small\texttt{ANY}}  \\
     $c$ & {\small\texttt{210.4.214.0/24}} & {\small\texttt{ANY}}  \\
     $d$ & {\small\texttt{210.4.215.0/24}} & {\small\texttt{ANY}}  \\
     $e$ & {\small\texttt{210.4.214.0/24}} & {\small\texttt{ICMP}} \\
     $\mathbf{f}$ & {\small\textbf{\texttt{0.0.0.0/0}}}      & {\small\textbf{\texttt{ICMP}}} \\
     $\mathbf{g}$ & {\small\textbf{\texttt{210.4.214.0/23}}} & {\small\textbf{\texttt{ICMP}}} \\
     $\mathbf{h}$ & {\small\textbf{\texttt{210.4.215.0/24}}} & {\small\textbf{\texttt{ICMP}}} \\
  \end{tabular}
  \caption{}
  \label{fig:second-example-c}
  \end{subfigure}
  \caption{Two meet-semilattices (\subref{fig:second-example-a} and \subref{fig:second-example-b}) for different subsets of (\subref{fig:second-example-c}) 2-dimensional match conditions}\label{fig:reannz-false-alarm}
\end{figure}

\begin{example}
\label{ex:algorithm}
Consider the DAG in~\Cref{fig:second-example-a} whose nodes correspond to the elements $a$ through $e$ in~\Cref{fig:second-example-c}, and suppose we want to insert element $\mathbf{f}$ now. For clarity in what follows, let $n_p$ denote a DAG node that satisfies $n_p.\mathit{elem} = p$. As expected, the call $\Call{Insert}{\mathbf{f}}$ creates DAG node $n_\mathbf{f}$. However, it also creates $n_\mathbf{g}$ and $n_\mathbf{h}$ for elements $\mathbf{g}$ and $\mathbf{h}$ in~\Cref{fig:second-example-c}, respectively, since $b \cap \mathbf{f} = \mathbf{g}$ and $d \cap \mathbf{g} = \mathbf{h}$; hence, $\mathit{Modified\_Nodes} = \{n_a, n_b, n_d, n_\mathbf{f}, n_\mathbf{g}, n_\mathbf{h}\}$. The resulting DAG is shown in~\Cref{fig:second-example-b} (new elements shown in bold). \qed
%, based on which~\Cref{alg:insert-node} re-computes PEC-cardinalities using~\Cref{alg:compute-cardinality}.
\end{example}

We compute \gls*{pec}-cardinalities using the set of modified nodes: once~\Cref{alg:insert-node} returns, the computation continues on \cref{line:dependency-begin} in~\Cref{alg:compute-cardinality} where \textsc{Compute\_Cardinality} is called for every node in $\mathit{Modified\_Nodes}$ (\cref{line:for-loop-modified-nodes-begin}--\ref{line:for-loop-modified-nodes-end}), and as it does so $\mathit{Modified\_Nodes}$ shrinks after each call to \textsc{Compute\_Cardinality}, until it becomes empty.

\begin{algorithm}[t]
  \caption{Compute and/or update cardinality of \glspl*{pec}}\label{alg:compute-cardinality}
  \begin{algorithmic}[1]
 % \Require $\mathit{Modified\_Nodes}$ is a global variable
  \Procedure{Compute\_Cardinality}{$\dagnode$}
    \State $\mathit{queue} \gets [\dagnode]; \mathit{visited} \gets \{\dagnode\}$\label{line:queue-init}\label{line:visited-init}
    %\State $\mathit{visited} \gets  \{\dagnode\}$\label{line:visited-init}
    \State $\dagnode.\mathit{cardinality} \gets \mathsf{cardinality}(\dagnode.\mathit{elem})$\label{line:base-case}
    \While{$\mathit{queue}$ is \textbf{not} empty}
       \State $\dagnode' \gets \mathit{queue}.\mathsf{dequeue}()$
       \For{$\mathit{child} \in \dagnode'.\mathit{children}$}
         \If{$\mathit{child} \not\in \mathit{visited}$}\label{line:visited-if}
           \If{$\mathit{child} \in \mathit{Modified\_Nodes}$}\label{line:dependency-begin}
             \State \Call{Compute\_Cardinality}{$\mathit{child}$}\label{line:compute-cardinality-of-child}
           \EndIf\label{line:dependency-end}
           \State $\mathit{visited}.\mathsf{insert}(\mathit{child})$
           \State $\dagnode.\mathit{cardinality} \gets$
           \State $\qquad \dagnode.\mathit{cardinality} - child.\mathit{cardinality}$\label{line:subtract}
           \State $\mathit{queue}.\mathsf{enqueue}(child)$
         \EndIf
       \EndFor
    \EndWhile
    \State $\mathit{Modified\_Nodes}.\mathsf{erase}(\dagnode)$
  \EndProcedure
  \end{algorithmic}
\end{algorithm}

The re-computation of \gls*{pec}-cardinalities works as follows. \textsc{Compute\_Cardinality} in~\Cref{alg:compute-cardinality} traverses the DAG using a queue (\cref{line:queue-init}). We initialize the \gls*{pec}-cardinality of the input DAG node $n$ by counting the packets in its associated element (\cref{line:base-case}), using the $\mathsf{cardinality}$ operator from~\cref{subsec:data-structures}, before subtracting the \gls*{pec}-cardinality of $n$'s descendants. To do so, the \gls*{pec}-cardinality of all the modified children is computed (\cref{line:compute-cardinality-of-child}) prior to updating $n$'s \gls*{pec}-cardinality (\cref{line:subtract}). Since there may be multiple paths to the same node in the DAG, \textsc{Compute\_Cardinality} uses a local variable (\cref{line:visited-init}) to ensure it does not subtract too much (\cref{line:visited-if}) as it traverses the sub-DAG rooted at \textsc{Compute\_Cardinality}'s input DAG node. By deferring the re-computation of \gls*{pec}-cardinalities for several insertions, the computation can be \emph{amortized}, if so desired. Note that the set of generated \glspl*{pec} is invariant under the insertion order of elements.

Unlike SAT/SMT solvers or BDD-based solutions---which can prove that a \glspl*{pec} is non-empty by finding a witness---\sharppec\ computes \gls*{pec}-cardinalities instead. In the worst case, this computation is quadratic in the size $N$ of the DAG where $N$ can be exponential in the number of input match conditions~\cite{BJMSV2016}.

%\begin{theorem}[Worst-case time complexity]\label{theorem:asymptotic-time-complexity}
%  \sharppec\ runs in $O(N^2)$ with $N = O(2^n)$ in the worst-case where $n$ is the number of input match conditions to be inserted.
%\end{theorem}

%We now explain the lattice-theoretical approach to network analysis in detail. For this, we have to go beyond IP prefixes, which are generally straightforward to analyze even without \glspl*{PEC}.

%Note that a local copy of $\mathit{Modified\_Nodes}$ is made (\cref{line:local-copy-modified-nodes}) prior to entering the loop because \textsc{Compute\_Cardinality} updates this global variable as it recursively traverses the DAG.

%\label{example:insert}
%The resulting meet-semilattice is illustrated by the portion of the Hasse diagram in~\Cref{fig:hasse-diagram} connected by solid lines. This smaller portion of the Hasse diagram includes all elements in~\Cref{fig:lattice-elements}, except $e$, $h$ and $j$. What happens when we insert $e$ into the DAG?
%\begin{figure}[b]
%    \[
%      \xymatrix@C=0.9em@R=0.4em{
%      &&\mathbf{a}\ar@{-}[d]& \\
%      &&\mathbf{b}\ar@{-}[ld]\ar@{-}[d]\ar@{.}[rd]& \\
%      &\mathbf{c}\ar@{-}[rd]&\mathbf{d}\ar@{-}[ld]\ar@{-}[d]\ar@{.}[rd]&e\ar@{.}[d] \\
%      &\mathbf{f}\ar@{-}[d]&\mathbf{g}\ar@{-}[ld]\ar@{.}[rd]&h\ar@{.}[d] \\
%      &\mathbf{i}\ar@{-}[rd]&&j\ar@{.}[ld] \\
%      &&\bot&&
%      }
%  \]
%\caption{The result of inserting element $e$ into the meet-semilattice that initially only contains the \textbf{bold} elements.}
%\end{figure}

\subsection{Answering Operator Questions via \gls*{pec}-based Queries}
\label{subsec:query}

When applying a \gls*{pec}-based formal network analysis technique to a set of packet headers decribed by a logical query, it is necessary to convert the query into a set of \glspl*{pec}. In \sharppec, we perform this conversion as follows.

Foremost, we assume that each logical query is a Boolean combination of logical predicates that have the same element type (\Cref{fig:element-type}) as the match conditions in the meet-semilattice. If a predicate in the query is not present in the meet-semilattice, we first insert it via~\Cref{alg:insert}.

Under these assumptions, a query is then converted into a set of \glspl*{pec} by invoking~\Cref{alg:query-to-pecs}, a recursive function on the logical structure of the input query. We remark that our last assumption ensures that \cref{line:find-node} in~\Cref{alg:query-to-pecs} always finds a node in the DAG, i.e., $n$ is never null. As part of our case study (\cref{subsec:case-study}), we give an example of a query conversion.

\begin{algorithm}
  \caption{Convert a query to a set of \glspl*{pec}}\label{alg:query-to-pecs}
  \begin{algorithmic}[1]
  \Function{Convert\_to\_PECs}{$\mathit{query}$}
    \If{$\mathit{query}\ \mathbf{is\ an\ element\ type}$}
      \State {$n \gets \Call{Find\_Node}{\mathit{query}}$}\label{line:find-node}
      \State \Return {$\Call{Subtree}{n}$}
    \ElsIf {$\exists\, g\,\colon\,\mathit{query} = \neg g$}
       \State $\mathit{Universe} \gets \Call{Subtree}{\mathit{Root}}$
       \State \Return {$\mathit{Universe} - \Call{Convert\_to\_PEC}{g}$}
     \ElsIf {$\exists\, g, h\,\colon\,\mathit{query} = g \wedge h$} \Comment{`$g \lor h$' is similar}
       \State $G \gets \Call{Convert\_to\_PECs}{g}$
       \State $H \gets \Call{Convert\_to\_PECs}{h}$
       \State \Return $G \cap H$ \Comment{Use `$\cup$' for `$g \lor h$'}
    \EndIf
  \EndFunction
  \end{algorithmic}
\end{algorithm}

\subsection{Minimality of \glspl*{pec}}
\label{subsec:minimality}

Next, we show that the set of non-empty \glspl*{pec} produced by \sharppec\ form atomic predicates in the following strict sense:

\begin{definition}[\textbf{Atomic Predicates~\cite{YL2013}}]
\label{def:atomic-predicates}
Let $\mathfrak{M}$ be a set of predicates, each of which represents a match condition of a firewall or forwarding rule. Then $\mathfrak{M}$'s set of atomic predicates, written $\mathsf{A}(\mathfrak{M}) = \{\alpha_1, \ldots, \alpha_k\}$, satisfies the following:
\begin{enumerate}
\item for all $i \in \{1, \ldots, k\}$, $\alpha_i \not= \mathbf{false}$;
\item $(\bigvee^k_{i=0} \alpha_i) = \mathbf{true}$;
\item $\alpha_i \wedge \alpha_j = \mathbf{false}$ for all $i, j \in \{1, \ldots, k\}$ such that $i \not= j$;
\item Each match condition $p$ in $\mathfrak{M}$, where $p \not= \mathbf{false}$, is equal to the disjunction of some subset of atomic predicates:
\begin{equation*}
p = \bigvee_{i \in S(p)} \alpha_i\ \text{where}\ S(p) \subseteq \{0, \ldots, k\};
\end{equation*}
\item $k$ is the minimal number such that the set $\mathsf{A}(\mathfrak{M}) = \{\alpha_1, \ldots, \alpha_k\}$ satisfies the above four conditions.
\end{enumerate}
\end{definition}

Yang and Lam show that the set of atomic predicates is unique~\cite{YL2013}, which they compute using BDDs. Given a set of rule match conditions that can be expressed using element types (\Cref{fig:element-type}), the next theorem shows how to compute this unique and minimal set through a fundamentally different algorithm that uses lattice theory and model counting.

\begin{theorem}[\textbf{Optimality of \sharppec}]\label{theorem:optimality}
Given as input a set of match conditions $\mathfrak{M}$ of an  element type, the set of non-empty \glspl*{pec} constructed by \sharppec\ forms atomic predicates $\mathsf{A}(\mathfrak{M})$.
\end{theorem}
\begin{proof}
The proof can be found in Appendix~\ref{appendix:proof-of-minimality}.
\end{proof}

Put simply, \sharppec's output is as good as APV's~\cite{YL2013}. We re-emphasize two important assumptions: (i)~match conditions must be expressed as element types (\Cref{fig:element-type}), and (ii)~the same inputs are supplied to both tools. Condition~(ii) is violated when, say, APV is allowed to pre-process forwarding rules by aggregating match conditions that are associated with the same output port. \sharppec\ does not perform such port-aggregation pre-processing step, a design decision we made because the partitioning of the packet header space would need to be re-computed every time some port is changed.
%If needed, the same result can be achieved by a post-processing step in \sharppec similar to the one descibed in~\cite{BJMSV2016}.

\section{Evaluation}
\label{sec:evaluation}

\begin{figure}[b]
\centering
  %\scalebox{0.92}{
  \small{
  \begin{tabular}{|l|l|} \hline
    \multicolumn{1}{|c|}{\small\textsc{Dataset}}              &
    \multicolumn{1}{c|}{\small\textsc{Short Description}}                        \\ \hline
REANNZ-IP~\cite{KARW2016,Cardigan}  & 1,159 distinct IP prefixes \\ \hline
REANNZ-Full~\cite{KARW2016,Cardigan} & 1,170 OpenFlow rules \\ \hline
Azure-DC~\cite{LBGJV2015}  & 2,942 ternary 128-bit vectors \\ \hline
Berkeley-IP~\cite{RouteViews,HKP2017} & 584,944 distinct IP prefixes \\ \hline
Stanford-IP~\cite{KVM2012}& 197,828 distinct IP prefixes \\ \hline
Stanford-Full~\cite{KVM2012} & 2,732 ternary 128-bit vectors\\ \hline
Diekmann~\cite{DMHC2016} & Thousands of 8-tuples \\ \hline
  \end{tabular}%
  }
  %}
      \caption{Summary of datasets}\label{fig:datasets}
\end{figure}

For our study, we experimentally evaluate different implementations of \gls*{pec}-construction schemes (\cref{subsec:implementation}), namely: \sharppec, APV, \ddnf\ and Veriflow where possible. We evaluate both the SAT/SMT and BDD-based solutions of the \gls*{pec}-emptiness problem as well as our counting method. As part of our evaluation, we analyze firewalls and forwarding rules collected from a variety of sources (\cref{sec:datasets}). Using these datasets, we uncover real-world cases where \ddnf\ raises false alarms and misses errors (\cref{subsec:case-study}), which \sharppec\ successfully avoids. We then evaluate \sharppec's performance (\cref{subsec:performance}).

\subsection{Implementations}
\label{subsec:implementation}

Here we outline our implementations of \gls*{pec}-emptiness checks (\cref{subsec:pec-emptiness-checks}), and APV as well as \sharppec\ (\cref{subsec:re-implementation-of-apv-and-ddnf}).

\subsubsection{Three \gls*{pec}-emptiness checking procedures}
\label{subsec:pec-emptiness-checks}

In addition to implementing \sharppec's counting method, we want to evaluate the SAT/SMT and BDD-based solutions to the \gls*{pec}-emptiness problem that use propositional logic to precisely encode when a \gls*{pec} is empty. Their symbolic encoding works as follows.

Let $x$ be a match condition of a rule, and denote with $C_x$ the set of direct children of element $x$ in the DAG constructed by \sharppec\ (recall~\cref{subsec:pecs-as-dags}). By construction, every child $c$ in $C_x$ is a strict subset of $x$, i.e., $c \subset x$. We emphasize that, for efficiency reasons, we only consider the direct children of $x$, so all children in $C_x$ are mutually incomparable, i.e., for any child $c$ and $c'$ in $C_x$, neither $c \subset c'$ nor $c' \subset c$. To implement the propositional logic \gls*{pec}-emptiness solutions, we construct the following Boolean formula: $x \wedge \neg \left(\bigvee_{c \in C_x}c\right)$---equivalently, $x \wedge \neg c_1 \wedge \neg c_2 \wedge \ldots \wedge \neg c_n$ where $c_1$, $c_2$, $\ldots$, $c_n$ are $x$'s direct children in $C_x$. For checking the formula's satisfiability, we use a SAT/SMT solver or construct a BDD, as detailed next.

Our BDD implementation uses the C++ BuDDy library. We set the intial node number and cache size by manual tuning and choosing values that yield better results.  In the case of the SAT/SMT implementation, we call Z3~\cite{DB2008}. To avoid additional parsing overhead, we use Z3's C++ API to construct the Boolean formulas, rather than using the more standard SMT-LIB~\cite{SMTLIB} or DIMACS format for SAT solvers.

As part of the Boolean encoding of element types (recall~\cref{subsec:element-types}), we convert \texttt{tbv<N>} elements into \texttt{N} Boolean variables, one for each non-$\ast$ ternary bit.
For the conversion of \texttt{set<T>}, which is implemented using bitsets, we encode the disjunction of the indexes of the set bits using $\lceil \log_2 K \rceil$ Boolean variables where $K$ is the length of the bitset. For the {\texttt{tuple<E$_1, \ldots, $E$_k$>}} encoding, we designate $b = b_1 + \ldots + b_k$ Boolean variables where $b_j$ is the number of Boolean variables needed to represent \texttt{E}$_j$. The final encoding is the conjunction of the Boolean encoding of each tuple coordinate.

\subsubsection{Implementation of APV, \ddnf\ and \sharppec}
\label{subsec:re-implementation-of-apv-and-ddnf}

To rigorously evaluate the performance of our tool against others, we implement a version of APV and \sharppec\ within the same framework: we opted for Z3~\cite{DB2008}. Our re-implementation of APV applies the same optimizations as proposed in~\cite{BJMSV2016}. We do not have to re-implement \ddnf, since it is already available as an open-source module in Z3. Similar to \ddnf, our implementation of \sharppec\ leverages Z3's highly optimized TBV implementation. We implement the other element types as a C++11 library, which we describe in more detail in Appendix~\ref{appendix:element-type}.

\subsection{Datasets}
\label{sec:datasets}

Our evaluation uses 64 different datasets, extracted from five independently collected routing tables and firewalls collections~\cite{KVM2012,LBGJV2015,KARW2016,DMHC2016,HKP2017}. \Cref{fig:datasets} categorizes our datasets according to their source of origin. Each dataset is encoded as a list of rule match conditions of a specific element type (recall~\Cref{fig:element-type} in~\cref{subsec:element-types}). Since \ddnf\ only supports TBVs, we encode the match conditions in our datasets as TBVs whenever possible, which is feasible for all datasets except the `Diekmann' dataset, as described below. Irrespective of the element type, we ensure all match conditions per dataset are unique, since duplicates could be processed with almost zero cost. We describe each category of dataset in turn.
%\vspace{-0.3em}
\paragraph{REANNZ} The REANNZ-Full dataset~\cite{KARW2016} contains more than a thousand OpenFlow rules, extracted from a single routing table that was used in the Cardigan deployment~\cite{Cardigan}. The OpenFlow rules in the REANNZ dataset use the following header fields: source and destination MAC addresses, ether-type, source and destination IP address, IP protocol field, and source and destination TCP ports. We convert each match condition in the rules to a ternary 216-bit vector. From the full dataset, we extract REANNZ-IP which contains only IP prefixes, but also encoded as TBVs.
%\vspace{-0.3em}
\paragraph{Berkeley-IP} The Berkeley-IP dataset originates from~\cite{HKP2017} where IPv4 prefixs from the RouteViews project~\cite{RouteViews} were evaluated in the context of the UC Berkeley campus network topology. Our dataset focuses only on the IPv4 prefixes, which we encode as 32-bit long TBVs.
%\vspace{-0.3em}
\paragraph{Azure-DC} The Azure-DC dataset~\cite{LBGJV2015} contains synthetic FIBs that simulate Azure-like data centers as deployed by Microsoft at that time. It contains a total of nearly 3000 match conditions, each of which is a ternary 128-bit vector.
%\vspace{-0.3em}
\paragraph{Stanford} The Stanford dataset originates from Stanford's backbone network~\cite{KVM2012}, which contains configurations of sixteen Cisco routers. For each router, we generate its transfer function~\cite{KVM2012} which models the static behavior of the router (including forwarding and ACLs). We then use the match conditions in the transfer function, encoded as ternary 128-bit vectors, to produce a dataset for that router (e.g Stanford-Full/boza). To measure the effect of analyzing a network containing all sixteen routers, we also combine all sixteen datasets into a single one, Stanford-Full, which contains a total of 2,732 unique ternary 128-bit vectors. In our Stanford-IP dataset, we extract the IP prefixes directly from the raw router configurations, thereby avoiding the IP prefix compression feature in HSA's transfer functions. As a result, our Stanford-IP datasets are significantly larger than the datasets used in the evaluation of HSA~\cite{KVM2012} and \ddnf~\cite{BJMSV2016}.
%\vspace{-0.3em}
\paragraph{Diekmann} The Diekmann datasets contains match conditions from real-world Linux iptables rule-sets~\cite{DMHC2016}. We parse the following packet header matching fields: source and destination IP prefix, source and destination port, protocol, connection state, input and output interface. We encode these as a mixture of TBVs and regular bitsets, which we combine into 8-tuples. We ignore wildcard characters for interfaces. We simplify each original iptables rule-set through a pre-processor that propagates match conditions along iptables chains in a depth-first manner, similar to function inlining. This essentially flattens a multi-chain iptables configuration into a list of match conditions without jumps and returns, so they conform to the same format as the other datasets.

\subsection{Case Study}
\label{subsec:case-study}

In this subsection, we describe real-world cases of imprecision in \ddnf, all of which \sharppec\ handles successfully. Due to space, we only illustrate a few examples (in our full study, we encountered over three dozen cases of imprecision in \ddnf).

To begin with, \ddnf\ misses 35 shadowed rules in the REANNZ dataset. We found that \ddnf\ misses four shadowed rules in the Stanford datasets, one in each of the `soza', `sozb', `yoza', and `yozb' Cisco routers. Furthermore, in the Stanford dataset, \ddnf\ fails to check that every packet whose destination IP address matches the IP prefix \texttt{171.64.79.160/24} is forwarded from router `yozb' to router `yoza'. For this query, \ddnf\ wrongly reports that some packets with such a destination IP address are dropped. The slightly simplified relevant rules in the dataset for the `yozb' router are as follows:

\begin{Verbatim}[fontsize=\small]
   Destination=171.64.79.160/28 => yoza
   Destination=171.64.79.176/28 => yoza
   Destination=171.64.79.128/27 => yoza
   Destination=171.64.79.192/27 => yoza
   Destination=171.64.79.224/27 => yoza
   Destination=171.64.79.0/25   => yoza
   Destination=171.64.79.0/24   => DROP
\end{Verbatim}

Here, \ddnf\ produces this wrong result, because the union of IP prefixes that forward to `yoza` equals the IP prefix of the last rule that drops packets: the match condition of the last rule, therefore, is encoded as a singleton set that contains an empty \gls*{pec}---the same underlying cause as described in~\cref{subsec:precision}.

As a more complicated example, consider the following human-readable form of the OpenFlow rules part in the REANNZ dataset (slightly simplified to help with readability), ordered from highest to lowest priority:

\begin{Verbatim}[fontsize=\small]
    Protocol=ICMP              => Controller
    Destination=210.4.214.0/24 => Port 1
    Destination=210.4.215.0/24 => Port 1
    Destination=210.4.214.0/23 => Port 2
    Destination=ANY            => DROP
\end{Verbatim}

The match conditions associated with these OpenFlow rules induce the DAG shown in~\Cref{fig:reannz-false-alarm}. Suppose a network operator wants to answer the following query:
\begin{quote}
``Are all non-\texttt{ICMP} packets destined to IP prefix \texttt{210.4.214.0/23} sent to Port 1?''
\end{quote}

Formally, this query is a Boolean combination of the form $\texttt{210.4.214.0/23},\texttt{ANY} \land \neg( \texttt{0.0.0.0/0},\texttt{ICMP})$ where the first and second conjunct are elements $b$ and $c$ in the DAG in~\Cref{fig:reannz-false-alarm}, respectively. Using~\Cref{alg:query-to-pecs} in~\cref{subsec:query}, we convert the query into the set of \glspl*{pec} $\{B,C,D\}$. Since $B$ is a \gls*{pec} associated with a rule that outputs the packet at port 2, \ddnf\ concludes that the above property is violated. However, \ddnf's verification result is incorrect: since $B$ is empty no such violation can be realized in the actual network. \sharppec\ correctly detects that the property holds. For the sake of brevity, we omit the discussion of five other, but similar, examples of imprecision in the REANNZ dataset.

%We argue that a network analysis cannot be formal if it exhibits the kind of problems as illustrate above.

%\begin{remark}\label{remark:false-alarm-construction}
%In general, the source of imprecision can be characterized as follows: let $X$ and $Y$ be two \glspl*{pec} where $X$ is empty and $Y$ is not, and let $r_X$ and $r_Y$ be two rules associated with $X$ and $Y$, respectively. If $r_X$ and $r_Y$ have different actions, and the query includes both \glspl*{pec}, then \ddnf's result is wrong when the property is true about $X$ and false about $Y$.
%\end{remark}

\subsection{Performance Evaluation}
\label{subsec:performance}

\begin{figure*}[t]
\centering
\scalebox{0.65}{
\begin{tabular}{@{}lllllllllllllll@{}}
\toprule
Dataset            & Insertions & PECs      & \begin{tabular}[c]{@{}l@{}}Empty\\ PECs\end{tabular} & \begin{tabular}[c]{@{}l@{}}Atomic\\ Preds.\end{tabular} & \multicolumn{2}{l|}{\begin{tabular}[c]{@{}l@{}}PEC-construction \\ time (s)\end{tabular}} & \multicolumn{3}{l}{PEC-emptiness check (s)} & APV (s) & \multicolumn{4}{c}{Memory (MB)} \\ \cmidrule(lr){6-10} \cmidrule(l){12-15}
                   &            &           &                                                      &                                                         & Z3 ddNF                           & \multicolumn{1}{l|}{\#PEC}                    & BDD       & SAT       & Card.              &         & BDD   & SAT   & Card. & APV     \\ \midrule
REANNZ-IP          & 1,159      & 1,160     & 25                                                   & 1,135                                                   & \textless{}1ms                    & \textless{}1ms                                & 0.016     & 0.414     & \textless{}1ms     & 0.001   & 6     & 6     & 3     & 5       \\
REANNZ-Full        & 1,170      & 12,783    & 275                                                  & 12,508                                                  & 0.112                             & 0.009                                         & 2         & 9         & 0.018              & 3       & 14    & 26    & 9     & 10      \\
Azure-DC           & 2,942      & 5,096,869 & 10,450                                               & 5,086,419                                               & 3301                              & 121                                           & 20112     & 47829     & 30                 & 25669   & 4,429 & 5,797 & 2,365 & 2,517   \\
Berkeley-IP        & 584,944    & 584,945   & 29,813                                               & Timeout                                                 & Timeout                           & 2709                                          & 1553      & 460       & 0.515              & Timeout & 302   & 701   & 227   & Timeout \\
Stanford-IP/soza   & 184,682    & 184,682   & 4,841                                                & 179,841                                                 & 471                               & 347                                           & 7         & 82        & 0.119              & 4951    & 102   & 251   & 69    & 49      \\
Stanford-IP/yoza   & 4,746      & 4,746     & 3                                                    & 4,743                                                   & \textless{}1ms                    & \textless{}1ms                                & 0.076     & 2         & 0.002              & 2       & 8     & 9     & 4     & 6       \\
Stanford-IP/All    & 197,828    & 197,828   & 4,874                                                & 192,954                                                 & 266                               & 199                                           & 19        & 89        & 0.156              & 5149    & 122   & 265   & 85    & 53      \\
Stanford-Full/soza & 524        & 16,764    & 81                                                   & 16,683                                                  & 0.056                             & \textless{}1ms                                & 0.668     & 9         & 0.024              & 2       & 18    & 19    & 10    & 13      \\
Stanford-Full/yoza & 507        & 60,363    & 231                                                  & 60,132                                                  & 5                                 & 0.17                                          & 4         & 38        & 0.17               & 20      & 46    & 65    & 31    & 28      \\
Stanford-Full/All  & 2,732      & 1,176,095 & 48,906                                               & 1,127,189                                               & 560                               & 28                                            & 692       & 1958      & 4                  & 2314    & 895   & 1,077 & 544   & 439     \\ \midrule
Diekmann/G & 5,321      & 889,646   & 40                                                   & 889,606                                                  & - & 39                                                           & 413       & 4729      & 10                 & 2385    & 3,843  & 3,854  & 3,924  & 608  \\
Diekmann/J & 6,004      & 1,058,897 & 56                                                   & 1,058,841                                                & - & 71                                                           & 486       & 5654      & 13                 & 2936    & 4,558  & 4,573  & 4,656  & 700  \\
Diekmann/K & 3,242      & 400,911   & 257                                                  & 400,654                                                  & - & 18                                                           & 157       & 2084      & 3                  & 732     & 1,997  & 2,006  & 2,031  & 233  \\
Diekmann/P & 578        & 492,378   & 4                                                    & 492,374                                                  & - & 47                                                           & 168       & 1837      & 4                  & 635     & 1,563  & 1,573  & 1,606  & 324  \\
Diekmann/Q & 307        & 4,626     & 38                                                   & 4,588                                                    & - & 0.087                                                        & 0.763     & 17        & 0.016              & 0.94    & 21     & 29     & 18     & 7    \\
\bottomrule
\end{tabular}
}
\caption{Evaluation results for a subset of datasets. See~Appendix~\ref{appendix:evaluation} for full experimental results table.}\label{fig:evaluation-results}
\end{figure*}

 We evaluate \sharppec's performance along two dimensions, namely: (i)~time and memory usage to construct \sharppec's meet-semilattice; (ii)~time and memory usage for detecting empty PECs. We discuss our results in turn.\footnote{All experiments are run on a Linux machine with an Intel Xenon CPU ES-1660 3.30GHz and 32GB DDR3 1333MHz RAM.}
 
 %\Cref{fig:evaluation-results}~and~\ref{fig:diekmann-evaluation-results} in~\Cref{evaluation} give the full details of all our experimental results.}

\subsubsection{PEC-construction}
We compare \sharppec\ to APV, and Z3's implementation of \ddnf. We ensure that every implementation benefits from the same optimizations (\Cref{subsec:re-implementation-of-apv-and-ddnf}). We find that \sharppec\ consistently outperforms APV and \ddnf\ in Z3 where, on larger datasets, the speed-up is more than $10\times$. For example, on the Azure-DC dataset, our re-implementation of \sharppec\ in Z3 is approximately $30\times$ faster than \ddnf. APV times out on the Berkeley-IP dataset after 10 hours, whereas \sharppec\ completes the PEC-construction in 45 minutes. We include in \sharppec's total run-time the time it takes to check PEC-emptiness, when comparing \sharppec\ and APV. For this comparison, we use the 39 datasets in which either APV or \sharppec\ runs for more than $100\,\mathit{ms}$, excluding the Berkeley-IP dataset where APV times out. In $95\%$ of these 39 cases, despite \sharppec's PEC-emptiness check, \sharppec\ is at least $10\times$ faster than APV, and $25\%$ of this time \sharppec's speed-up is at least $100\times$. On average, \sharppec\ is at least $80\times$ faster than APV. Finally, APV and \sharppec's memory usage averages out to be the same across these datasets. \Cref{fig:evaluation-results} shows parts our experimental results, see Appendix~\ref{appendix:evaluation} for the full details.

The fact that \sharppec\ outperforms APV is expected, since \sharppec\ eliminates the per-bit overhead of BDDs. The performance difference between \sharppec\ and Z3's implementation of \ddnf, in turn, can be explained in terms of the number of intersection and subset operations required to insert a new match condition into their respective data structure: their total run-time is proportional to these operations. For example, in the Stanford-Full dataset, \sharppec\ requires 0.4 million whereas \ddnf\ in Z3 takes 8 million such operations, a $20\times$ improvement. \sharppec's improvement over Z3's implementation of \ddnf\ are similar on the other datasets.

\subsubsection{PEC-emptiness checking}
\label{subsec:pec-emptiness-checking-evaluation}
We compare \sharppec's counting method to the SAT/SMT and BDD-based solutions to checking PEC-emptiness. We evaluate the performance of PEC-emptiness checking using the 24 datasets in which \sharppec\ runs for more than $100\,\mathit{ms}$. We perform the PEC-emptiness check after the PEC-construction has completed. We take extra precautions in our implementations to ensure a fair comparison (\cref{subsec:pec-emptiness-checks}). \Cref{fig:evaluation-results} shows that \sharppec's counting method significantly outperforms the SAT/SMT and BDD-based approaches: \sharppec\ achieves at least a $10\times$ speed-up compared to the SAT/SMT and BDD-based approach in over $95\%$ of cases. On average, \sharppec\ is at least $500\times$ and $200\times$ faster than the SAT/SMT and BDD-based approaches, respectively.

%maybe move this paragraph to previous sections when we talk about the cardinality based approach?
To understand why \sharppec's cardinality-based approach outperforms the SAT/SMT and the BDD-based approaches, reconsider the IP prefixes in~\Cref{subsec:precision}.
Representing $x$, $y$, and $z$ in propositional logic requires 19, 19, and 18 variable assignments respectively, corresponding to their non-wildcard bits. Just encoding $Z = z - (x \cup y)$ in SAT requires near 60 logic gates, excluding the task of checking satisfiability.
Representing the predicates using BDDs requires the same number of BDD nodes. Assuming logical BDD operations are linear in their operand size, computing $Z$ at least requires CPU cycles proportional to the cumulative size of the three BDDs.
On the other hand, the cardinality of each predicate in the example fits into a single machine word. We need only 2 arithmetic CPU operations to compute the cardinality of $Z$ (i.e $|z| - |x| - |y|)$, and then check if it is zero. While in theory there are still near 60 operations performed (at the bit level), \sharppec\ harnesses the computing power of ALUs to finish the operations in fewer CPU cycles. For example, in the Stanford-Full dataset where each node in the DAG has 3 children and 12 nodes in its subtree on average, the BDD-based approach requires $3\times128$ low-level BDD operations on average (each spanning tens of CPU instructions). By contrast, our cardinality-based approach needs at most 3 ALU operations for each subtraction. So \sharppec\ should be at least  $(3\times128) / (12\times3) \approx 10\times$ faster than the BDD-based approach, and our experiments show indeed at least a $127\times$ speed-up.

%just some cosmetic changes. You can undo if you want.
%The astute reader will notice that computing the cardinality of a DAG node $n$ requires the traversal of $n$'s entire subtree, whereas SAT/SMT and BDD-based \gls*{pec}-emptiness checks need to only consider $n$'s direct children. Our experiments show that this trade-off is justified: the cardinality-based approach is still faster. For example, in the Stanford-Full dataset, each node in the DAG has 3 children and 12 nodes in its subtree on average. The dataset is encoded using 128-bit TBVs. In the BDD-based approach, each node requires $3\times128$ low-level BDD operations on average (each spanning tens of CPU instructions). By contrast, the cardinality-based approach needs at most 3 ALU operations for each subtraction. So \sharppec\ should be at least  $(3\times128) / (12\times3) \approx 10\times$ faster than the BDD-based approach, and our experiments show indeed at least a $127\times$ speed-up (\Cref{fig:evaluation-results} in~\Cref{appendix:evaluation}).

\subsubsection{Comparison with Veriflow} We compare \sharppec\ to the original implementation of Veriflow~\cite{AhmedAndBrighten}. Since that implementation of Veriflow only supports a restricted form of OpenFlow rules where arbitrary per-field bitmasks are disallowed, it cannot analyze the majority of our datasets. We therefore restrict our experiments with Veriflow to a simplified version of the Stanford-Full dataset. We use the default packet header field ordering in Veriflow. We ask Veriflow to only find `Equivalence Classes' (ECs), rather than each EC's forwarding graph. In this restricted setting, Veriflow takes $41\,s$ to create 3,778,324 ECs, using $1\,\mathit{GB}$ of memory. Despite \sharppec's support for arbitrary bitmasks, it is still more efficient than Veriflow, in both time ($30\,s$) and space ($0.5\,\mathit{GB}$): specifically, \sharppec\ constructs only 1,066,645 PECs in $27\,s$, and finds 44,418 empty PECs in 3s. 

%I agree that there needs to be some conclusive thought, but this sentence currently sounds speculative:
%The excessive number of ECs produced by Veriflow (nearly $4\times$ \sharppec's non-empty PECs) would subsequently make the network analysis proportionally slower.

%\AH{Are there more data points? This is currently reading a bit like a ``cliff hanger'' :)}
%\AK{there are data points for other other stanford datasets (each router individually). Also the IP datasets. Maybe reannaz too}

%TODO: also compare with Delta-net maybe?
% also maybe mention that HSA is irrelevant

\subsection{Discussion: Importance of  Empty \glspl*{pec}} We showed that \ddnf's non-minimality of \glspl*{pec} is due to \glspl*{pec} that are empty. In our case study (\Cref{subsec:case-study}), we exemplified real-word cases where empty \glspl*{pec} lead to wrong analysis results, which are very likely to hinder technology adoption~\cite{SAEMJ2018}. We emphasize that we only gave illustrative examples; our list is not exhaustive, and it includes cases where \ddnf\ misses errors. In practice, therefore, \ddnf\ is only as fast as the slowest decision procedure needed to sanity check its results, a fundamental limitation. By contrast, \sharppec's analysis is correct by construction (\Cref{subsec:minimality}), and its performance is \emph{not} dependent on BDDs or SAT/SMT solvers, which are orders of magnitude slower in finding empty \glspl*{pec} (\Cref{subsec:pec-emptiness-checking-evaluation}).

%In fact, wrong analysis results arise under relatively weak conditions (\Cref{remark:false-alarm-construction}), which we found to be readily satisfied by Stanford, REANNZ, and a subset of the Diekmann datasets.

\section{Related Work}
\label{sec:related-work}

%In this section, we discuss related work in the literature.

Similar to APV~\cite{YL2013} and \ddnf~\cite{BJMSV2016}, \sharppec\ has many potential applications in the field of network correctness. The literature in this field is vast and includes BGP configuration checking (e.g.,~\cite{PKGC2017,GW1999,FB2005,QU2005,FSR2012,Bagpipe2016,ERA2016,GVAM2016,PCKGC20}), ACL misconfiguration detection (e.g.,~\cite{BGR2011,aclabstractionrefinement2011}), firewall checking (e.g.,~\cite{YMSCCM2006,JS2009,NBDFK2010,ZMMN2012}), SDN verification (e.g.,~\cite{CVPDR2012,SSYPG2012,BBGIKSSV2014,cocoon2017}), testing (e.g.,~\cite{ZKVM2012,taste2014,chaosmonkey2015,FYTCS2016,crystalnet}), debugging (e.g.,~\cite{BZZMRW2014,sts2015}), differential analysis (e.g.,~\cite{chimp2015}), concurrency analysis (e.g~\cite{sdnracer2015,bigbug2017}), automatic repair (e.g.,~\cite{marham2016,cpr2017,fixit2017}), synthesis (e.g.~\cite{propane2016,el2017network,BDCVV2018}), programming languages (e.g.~\cite{FHFMRSW2011,SGW2014,AFGJZSW2014,kinetic2015,p4k2018}), safe network updates (e.g.,~\cite{reitblatt2012abstractions,hotswap2013,safeupdatehybridsdn2013,ezseqway2017}), data plane checking (e.g.,~\cite{XZMZGHR2005,MKACGK2011,KVM2012,LBGJV2015}), real-time checkers~\cite{KCZVMcKW2013,KZZCG2013,ZZYJJLMV2014,HKP2017}, and more general network analyses (e.g.,~\cite{AMEE2009,AA2010,secguru2014,FFPWGMM2015,era,BGMW2017}) together with suitable levels of abstractions (e.g.,~\cite{PBLRV2016,BGMW2018}).

\begin{figure}[b]
\centering
    \scalebox{0.9}{
    \small{
  \begin{tabular}{|l|>{\centering\arraybackslash}p{0.3cm}|>{\centering\arraybackslash}p{0.3cm}|>{\centering\arraybackslash}p{0.3cm}|>{\centering\arraybackslash}p{0.3cm}|>{\centering\arraybackslash}p{0.3cm}|}
    \hline
    
	{\textbf{\textsc{Feature / Characteristic}}}              & \rotatebox{90}{\sharppec} & \rotatebox{90}{\ddnf} & \rotatebox{90}{\textsc{APV\ }} & \rotatebox{90}{\textsc{Delta-net\ }}  & \rotatebox{90}{\textsc{Veriflow\ }}  \\ \hline
	Precise network analysis           & $\CIRCLE$           & $\Circle$             & $\CIRCLE$      & $\CIRCLE$   & $\CIRCLE$   \\
	Rule priority \& action invariant                  & $\CIRCLE$           & $\CIRCLE$           & $\Circle$        & $\CIRCLE$   & $\CIRCLE$   \\ \hline
	Match conditions with bit masks                 & $\CIRCLE$           & $\CIRCLE$           & $\CIRCLE$      & $\Circle$     & $\LEFTcircle$   \\
	Wildcard on packet header fields                 & $\CIRCLE$           & $\CIRCLE$           & $\CIRCLE$      & $\Circle$     & $\CIRCLE$   \\
	Match conditions with sets of values     & $\CIRCLE$           & $\Circle$             & $\CIRCLE$      & $\Circle$     & $\Circle$   \\
	Negation on packet header fields     & $\CIRCLE$           & $\Circle$             & $\CIRCLE$      & $\Circle$     & $\Circle$   \\
	Range filters beyond IP prefixes     & $\CIRCLE$           & $\Circle$             & $\CIRCLE$      & $\CIRCLE$     & $\Circle$   \\ \hline
	\gls*{pec}-cardinalities        & $\CIRCLE$           & $\Circle$             & $\Circle$        & $\Circle$     & $\Circle$   \\
	Negation-free \gls*{pec}-construction & $\CIRCLE$           & $\CIRCLE$           & $\Circle$        & $\CIRCLE$   & $\Circle$   \\
	Canonical \gls*{pec}-representation                & $\CIRCLE$           & $\CIRCLE$           & $\CIRCLE$      & $\CIRCLE$   & $\Circle$   \\
 	Minimal and unique set of \glspl*{pec}                 & $\CIRCLE$             & $\Circle$             & $\CIRCLE$      & $\Circle$     & $\Circle$     \\ \hline
  \end{tabular}
    }
    }
  \caption{Feature comparison of closest related work}\label{fig:related-work}
\end{figure}

Our work is most closely related to \ddnf~\cite{BJMSV2016}, APV~\cite{YL2013}, Delta-net~\cite{HKP2017} and~Veriflow~\cite{KZZCG2013}, since they all partition packet headers somehow. However, these formal network analysis tools also differ in important ways, as summarized by~\Cref{fig:related-work} using characteristics, which are divided into three blocks: (i)~whether the analysis is precise \glspl*{pec} remain the same when the priority or output port of rules change; (ii)~common kinds of match conditions of practical interest; and (iii)~finally, attributes of the underlying algorithms. We discuss each of these tools in turn:%, including features not listed in~\Cref{fig:related-work}.
%\vspace{-0.3em}
\paragraph{\ddnf~\cite{BJMSV2016}}\label{sec:related-work-ddnf} \sharppec\ achieves precision when \ddnf\ cannot. Furthermore, we have shown that \sharppec\ can detect shadowed rules, whereas \ddnf\ cannot in general. As a result, \sharppec\ can verify equivalence of forwarding tables, whereas \ddnf\ cannot. We have also shown that \sharppec\ is more expressive than \ddnf\ in the kind of match conditions supported, e.g., iptables rule-sets. The DAG produced by \sharppec\ can be shown to be isomorphic to \ddnf's, but \sharppec\ is up to $30\times$ faster than \ddnf\ in constructing it (\cref{sec:evaluation}).
%\vspace{-0.3em}
\paragraph{APV~\cite{YL2013}} \label{sec:related-work-apv} APV produces \glspl*{pec} in the form of atomic predicates, the smallest partition of the packet header space. \sharppec\ also constructs the fewest \glspl*{pec} (\Cref{theorem:optimality}), and it does so $10\times$ faster than APV (\cref{subsec:performance}). Through an optional pre-processing step, APV may further reduce the problem size by aggregating match conditions per output port. However, when the priority or the output port of a rule changes, so would atomic predicates for the \emph{entire} network then. By contrast, \sharppec\ and \ddnf\ only create \glspl*{pec} that are invariant under changes to the priority of rules and/or their actions. As explained in the introduction, APV's \gls*{pec}-construction algorithm is not negation-free, explaining why it relies on BDDs~\cite{B1986,K2009}, whereas neither \sharppec\ nor \ddnf\ do.

%By contrast, neither \sharppec\ nor \ddnf\ use negation, thereby avoiding the performance bottleneck of BDDs.
%\vspace{-0.3em}
\paragraph{Delta-net~\cite{HKP2017}}\label{sec:related-work-delta-net} Delta-net is specifically designed for real-time analysis of large-scale BGP-controlled data centers~\cite{rfc7938}. It only supports forwarding rules that match packets based on ranges, possibly with arbitrary lower and upper bounds (unlike \ddnf). Due to its limited expressiveness, Delta-net achieves quasi-linear time complexity, whereas \ddnf\ and \sharppec's higher expressiveness has an exponential worst-case time complexity. In addition, unlike \ddnf\ and \sharppec, Delta-net's run-time is independent of the order in which match conditions are inserted. Delta-net exploits the fact that the negation of a range can be efficiently computed, so its \gls*{pec}-construction scheme is not negation-free.
%because the negation of an IP prefix may yield half-closed intervals that cannot be efficiently represented in an IP prefix list format.
%\vspace{-0.3em}
\paragraph{Veriflow~\cite{KZZCG2013}}\label{sec:related-work-veriflow} Veriflow uses a multi-dimensional trie data structure to represent \glspl*{pec}. To do so efficiently, Veriflow imposes assumption that prevent it from analyzing most multi-dimensional match conditions in our datasets (\cref{subsec:performance}). The \glspl*{pec} constructed by Veriflow depend on the order of levels in the multi-dimensional tries, which can render its memory usage and run-time performance unpredictable.

%Furthermore, negation is less efficient in Veriflow than APV since negation using BDDs is a constant-time operation~\cite{K2009}. In contrast to both, the partitioning schemes in \ddnf\ and \sharppec\ completely eliminate the cost of negation during \gls*{pec}-construction.

\section{Concluding Remarks}
\label{sec:conclusion}

%This section needs more work. I am not sure yet what to say and not to say.

%\paragraph{\textbf{Importance of detecting empty \glspl*{pec} in practice}:}{In our case study (\Cref{subsec:case-study}), we exemplified real-word cases where \ddnf\ produces wrong analysis results due to empty \glspl*{pec}. While we believe it is self-evident that any kind of wrong result in a formal analysis setting is undesirable (possibly harming technology adoption~\cite{SAEMJ2018}), we emphasize that we only gave illustrative examples; our list is not exhaustive. In fact, wrong analysis results arise under relatively weak conditions (\Cref{remark:false-alarm-construction}), which we found to be readily satisfied by Stanford, REANNZ, and a subset of the Diekmann datasets. By finding empty \glspl*{pec}, \sharppec's analysis is correct. In addition, we showed that by eliminating such \glspl*{pec}, \sharppec\ achieves minimality (\Cref{subsec:minimality}). Finally, ~\cite{BJMSV2016} shows how packet modification can be handled using pre-images in \ddnf. We believe that the ability to detect empty PECs enables us to handle packet modification using post-images as well, the details of which we leave as a future work.}

Our case study (\Cref{subsec:case-study}) and experiments (\cref{subsec:performance}) reveal the tension between precision, expressiveness and performance: Veriflow and \ddnf\ impose assumptions that prevent them from analyzing most of our dataset, and \ddnf's analysis is imprecise. By contrast, APV is very expressive and precise but significantly slower than Veriflow and \ddnf. Our work offers a new lattice-theoretical, algorithmic framework for formal network analysis that is expressive, precise and fast, thereby addressing a longstanding problem that has spanned three generations of formal network analysis tools.

To achieve this, we identified and efficiently solved the coNP-hard problem (Appendix~\ref{appendix:complexity}) of deciding whether a \gls*{pec} is empty or not. We showed that both SAT/SMT and BDD-based solutions to this problem perform poorly. This lead to \sharppec, which uses a model counting method that is $10-100\times$ faster than the SAT/SMT and BDD-based solutions. In addition, \sharppec\ constructs the unique minimal number of \glspl*{pec}, and it does so $10\times$ faster than APV's atomic predicates.

\bibliographystyle{IEEEtran}
\bibliography{sharppec}

\begin{appendices}

\newpage
~
\newpage

\section{Worst-case Complexity}
\label{appendix:complexity}

Here we prove results about the theoretical worst-case complexity of \sharppec's underlying model counting method.

Given a pair $\langle x,\,Y \rangle$ where $x$ is an element (recall~\cref{subsec:data-structures}) and $Y$ is a finite set of elements, define $\mathsf{PEC}{\langle x,\,Y \rangle} \triangleq x - \bigcup_{y \in Y} y$ to be the set of packet headers in $x$ that are not included in the union of elements in $Y$. Without loss of generality, we make the simplifying assumptions that all elements in $Y$ are a subset of $x$, and every element in $Y$ is maximal with respect to subset inclusion. Both conditions can be satisfied through a preprocessor that runs in polynomial time. For example, if $Y$ contains the ternary bit vectors `$101$' and `$10\ast$', we can remove the former because $101 \subseteq 10\ast$.

Define the \gls*{pec}-emptiness problem as follows: given a pair $\langle x,\,Y \rangle$ where $x$ is an element and $Y$ is a set of elements, is $\mathsf{PEC}{\langle x,\,Y \rangle}$ empty? In an analogous way, \gls*{pec}-cardinality is defined to be the number of packet headers in $\mathsf{PEC}{\langle x,\,Y \rangle}$.

\begin{figure}[b]
  \centering
    \[
      \xymatrix@C=0.2em@R=2em{
                         &          &    \top\,\colon\,\mathbf{4}  &           &                   \\
      1\ast1\ar@{-}[urr]\,\colon\,\mathbf{1} &          & 11\ast\ar@{-}[u]\,\colon\,\mathbf{0} &           & \ast10\ar@{-}[ull]\,\colon\,\mathbf{1} \\
                         &  111\ar@{-}[ul]\ar@{-}[ur]\,\colon\,\mathbf{1} & & 110\ar@{-}[ul]\ar@{-}[ur]\,\colon\,\mathbf{1} &    \\
                         &          &  \bot\ar@{-}[ul]\ar@{-}[ur]\,\colon\,\mathbf{0}  &               &
      }
  \]
  \caption{Hasse diagram and \gls*{pec}-cardinalities (in \textbf{bold}) for the DNF formula $(x_1 \wedge x_3) \vee (x_1 \wedge x_2) \vee (x_2 \wedge \neg x_3)$}
  \label{fig:dnf-hasse-diagram}
\end{figure}
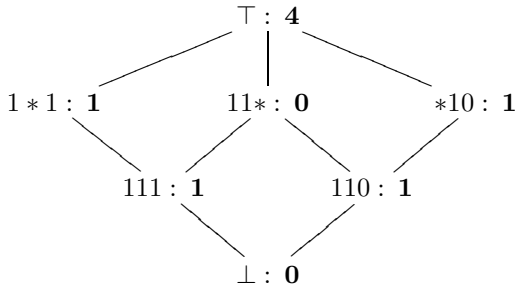

\begin{theorem}[\gls*{pec}-emptiness complexity]\label{theorem:np-complete-complexity}
  The problem of deciding whether a \gls*{pec} is empty or not is coNP-complete.
\end{theorem}
\begin{proof}
  To show that the \gls*{pec}-emptiness problem is coNP, it suffices to prove that deciding the non-emptiness of a \gls*{pec} is NP, which requires only a packet header as witness. To prove coNP-hardness, we proceed as follows.

  Let \textsc{TAUT} be the problem of deciding whether a given propositional logic formula in \emph{disjunctive normal form} (\emph{DNF}) is a tautology, a coNP-complete problem. To show that the \gls*{pec}-emptiness problem is coNP-complete, we reduce from \textsc{TAUT}. For this reduction, recall that a DNF formula is a disjunction of \emph{clauses}, each of which is a conjunction of \emph{literals} of the form $x$ or $\neg x$ for some Boolean variable $x$. Let $\phi$ be a DNF formula over the Boolean variables $x_1$, $x_2$, \ldots, $x_n$. For each clause $c_k$ in $\phi$, we can construct in polynomial time a ternary bit-vector $v_k = \langle y_1, y_2, \ldots, y_n \rangle$ where each ternary bit $y_i$ for $1 \leq i \leq n$ satisfies the following:
  \begin{displaymath}
    y_i = \begin{cases}
      1     &\mbox{if } \ell_i \in \phi \mbox{ and } \ell_i = x_i, \\
      0     &\mbox{if } \ell_i \in \phi \mbox{ and } \ell_i = \neg x_i, \\
      \ast  &\mbox{if } \phi \mbox{ contains no literal } \ell_i.
    \end{cases}
  \end{displaymath}
  Let $V$ be the set of ternary bit vectors induced by $\phi$, and define $\top$ to be the ternary bit vector of length $n$ where each bit is a wildcard, i.e., `$\ast$'. Then $\mathsf{PEC}{\langle \top,\,V \rangle}$ is empty if and only if $\phi$ is a tautology, proving coNP-hardness. We conclude that the \gls*{pec}-emptiness problem is coNP-complete.
\end{proof}

\begin{example}\label{example:np-complete-constr}
To illustrate the reduction in~\Cref{theorem:np-complete-complexity}, consider the following DNF formula $\phi$: $(x_1 \wedge x_3) \vee (x_1 \wedge x_2) \vee (x_2 \wedge \neg x_3)$. Each element $x$ in the associated Hasse diagram in~\Cref{fig:dnf-hasse-diagram} is also annotated with the cardinality of $\mathsf{PEC}{\langle x,\,C_x \rangle}$ where $C_x$ contains all direct children of $x$. Note that the \gls*{pec} that is associated with the top element ($\top$) is non-empty, which means that $\phi$ is not a tautology. We remark that the clause $x_1 \wedge x_2$ is redundant in $\phi$ in the sense that its removal does not change the truth values of $\phi$. This redundancy surfaces as an empty \gls*{pec} (namely, $\mathsf{PEC}{\langle 11\ast,\,\{111, 110\} \rangle}$), as denoted by `$11\ast\,\colon\,0$' in~\Cref{fig:dnf-hasse-diagram}. \qed
\end{example}

In the mode of operation where \sharppec\ counts packet headers in each \gls*{pec}, it is not difficult to see that the produced \gls*{pec}-cardinality information can answer the following \#P-hard counting problem (\#DNF): \emph{how many different variable assignments will satisfy a given formula in \emph{DNF}?} The following proof reduces \#DNF to the problem of counting the number of packet headers that are not matched by any of the input match conditions, proving the \#P-hardness of the \gls*{pec}-cardinality problem. The quantity in the last proof step is illustrated by the outermost gray area at outermost part of the Venn diagram in~\Cref{fig:venn}, where circles and differently colored regions denote input match conditions and \glspl*{pec}, respectively.

\begin{theorem}[Complexity of \gls*{pec}-cardinality]\label{theorem:complexity}
   Counting the packet headers in the disjunction of input match conditions is a \#P-hard problem.
\end{theorem}
\begin{proof}
The proof proceeds by reduction from \#DNF:
\begin{enumerate}
  \item on input of a DNF formula $\phi$ over $n$ Boolean variables, convert each clause $c_k$ in $\phi$ to a $n$-length ternary bit vector $v_k$, as defined in~\Cref{theorem:np-complete-complexity};
  \item collect these $n$-length ternary bit vectors into set $V$;
  \item send $V$ to oracle to obtain the cardinality of $\mathsf{PEC}{\langle \top,\,V \rangle}$;
  \item subtract $\mathsf{PEC}{\langle \top,\,V \rangle}$'s cardinality from $2^n$.
\end{enumerate}
\end{proof}

\begin{figure}[b]
  \centering
    \begin{tikzpicture}[framed,rounded corners,background rectangle/.style={draw=black,fill=lightgray}]
        \begin{scope}
          \fill[preaction={fill, white},pattern=crosshatch,pattern color=red]\firstcircle;
          \fill[preaction={fill, white},pattern=north east lines,pattern color=green] \secondcircle;
          \fill[preaction={fill, white},pattern=crosshatch dots,pattern color=blue] \thirdcircle;
        \end{scope}
        \begin{scope}
          \clip \secondcircle;
          \fill[yellow] \firstcircle;
        \end{scope}
        \begin{scope}
          \clip \secondcircle;
          \fill[cyan] \thirdcircle;
        \end{scope}
        \begin{scope}
          \clip \firstcircle;
          \fill[magenta] \thirdcircle;
        \end{scope}
        \begin{scope}
          \clip \firstcircle;
          \clip \secondcircle;
          \fill[white] \thirdcircle;
        \end{scope}
        %\node at (0.72,0.72) {$0$};
        \draw \firstcircle node[text=black,above] {};
        \draw \secondcircle node [text=black,below left] {};
        \draw \thirdcircle node [text=black,below right] {};
    \end{tikzpicture}
      \caption{Venn diagram of three match conditions, denoted by circles, which collectively induce eight \glspl*{pec}, shown as differently shaded regions}
     \label{fig:venn}
  \end{figure}
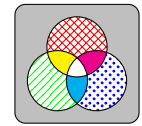%

\begin{example}
  We continue~\Cref{example:np-complete-constr}. Suppose the three clauses in the DNF formula $\phi$ represent match conditions. The number of packet headers in the disjunction of these match conditions is then $2^3 - 4 = 4$, since $\top\,\colon\,\mathbf{4}$ according to~\Cref{fig:dnf-hasse-diagram} where $\top = \ast\ast\ast$ matches any three bits. \qed
\end{example}

\section{Implementation Details}
\label{appendix:element-type}

In this appendix, we give more details regarding the implementation of element types. Since it is easy to implement {\texttt{ip\_prefix}}, {\texttt{optional<T>}} and {\texttt{tbv<N>}} as bit vectors, we focus our discussion on {\texttt{disjoint\_ranges}}, {\texttt{set<T>}} as well as \texttt{tuple}.

\begin{figure}
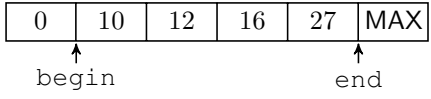

  \centering
  \disjointrangesfig{1}{5}{$0$
    \nodepart{two}$10$%
    \nodepart{three}$12$%
    \nodepart{four}$16$%
    \nodepart{five}$27$%
    \nodepart{six}$\mathsf{MAX}$%
   }%
  \caption{Internals of \texttt{disjoint\_ranges} where the $\mathsf{MAX}$ constant corresponds to the maximal upper bound of any range}\label{fig:disjoint-ranges}
\end{figure}

Firstly, {\texttt{disjoint\_ranges}} internally represents boundaries of half-closed intervals as a sorted array of numbers. An invariant of {\texttt{disjoint\_ranges}} is that the beginning and end of the underlying array contain the smallest (i.e., zero) and largest (i.e., $\mathsf{MAX}$) representable number, respectively. By adjusting an internal offsets for delimiting the array bounds, this ensures that negation on {\texttt{disjoint\_ranges}} is a constant-time operation, which allows us to efficiently represent complements of, say, arbitrary TCP/IP port range.

For example,~\Cref{fig:disjoint-ranges} shows the internals of the disjoint set of half-closed intervals $\{ [10:12), [16:27) \}$, whereas the set $\{ [0:10), [12:16), [27:\mathsf{MAX}) \}$ corresponds to its negation by adjusting the \texttt{begin} and \texttt{end} delimiters accordingly. As expected, the cardinality of {\texttt{disjoint\_ranges}} is merely the sum of its constituent ranges.

Similar to TBVs in \ddnf, we implement \texttt{set<T>} as a `bitset', i.e., a compact heap-allocated array of machine words which are manipulated via bitwise operators. As a result, set intersections, subset checks, cardinality computation (i.e., number of set bits), and negations are performed very efficiently. The length of \texttt{set<T>}'s underlying bitset is equal to the number of distinct values of \texttt{T} used by all the input sets plus an additional bit representing all values not explicitly used in the inputs (in case at least one such value exists).

By standard point-wise extension~\cite{DP2002}, all tuples constructed from the other element types in~\Cref{fig:element-type} also form partially ordered sets. For example, if $x$ and $y$ are 3-tuples that represent the match conditions in~\Cref{fig:lattice-elements}, then $x \subseteq y$ if and only if the following coordinate-wise subset inclusions hold: (i)~$\mathsf{source}(x) \subseteq \mathsf{source}(y)$, (ii)~$\mathsf{destination}(x) \subseteq \mathsf{destination}(y)$ and (iii)~$\mathsf{protocol}(x) \subseteq \mathsf{protocol}(y)$.

We remark that the cardinality operator on a $k$-tuple is computed by multiplying the cardinalities of all its $k$ elements, which can be done efficiently by using standard arbitrary-precision integers, where  folding expressions for tuples, as in C++17~\cite{CPP17}, allow for effective optimizations.

\section{Proof: Minimality of PECs}
\label{appendix:proof-of-minimality}

In this subsection, we give the proof of the `Optimality of \sharppec' theorem in~\Cref{subsec:minimality}:

\begin{proof}[Proof]
Recall that \#PEC efficiently detects all empty PECs using its counting method. It is easy to see that this set of non-empty PECs computed by \#PEC satisfies the first four conditions of the definition of atomic predicates (\Cref{def:atomic-predicates}). 

It remains to show that the set of non-empty PECs satisfies condition 5 (minimality). We write `YL' for Yang and Lam's Algorithm 3 in~\cite{YL2013}. Fix $\mathfrak{M}$ to be a set of predicates input to YL and \#PEC. Since YL generates atomic predicates~\cite{YL2013}, which are minimal by definition, it suffices to show that the predicates represented by every non-empty PEC produced by \#PEC is in the set of atomic predicates generated by YL.
%%%%insert algorithm 3 from APV paper here

By a simple induction on $n$, it is easy to show that for a set of input predicates $\mathfrak{M} = \{P_1, P_2, ..., P_n\}$, each predicate of the form $\phi = X_1 \land X_2 \land ... \land X_n \neq \bot$ where $X_i \in \{P_i, \neg P_i\}$ for $1 \leq i \leq n$ is in the set of atomic predicates generated by YL. Therefore, due to the uniqueness of non-empty PECs generated by \#PEC (a corollary of condition 3), it suffices to show that each non-empty PEC generated by \#PEC for input $\mathfrak{M}$ represents a predicate of the above conjunctive form.

Let $n$ be a DAG node whose cardinality is non-zero. This means that $n$ represents a PEC (i.e.,  $\mathsf{PEC}(n) \triangleq n.\mathit{elem} - \bigvee_{c \in n.\mathit{children}}{c.\mathit{elem}}$) that is non-empty. For each DAG node $m$, let $A(m)$ be the set of its ancestors (including $m$) in the DAG produced by \#PEC. We claim that the predicate represented by $\mathsf{PEC}(n)$ is equivalent to (written ``$\equiv$'') $\bigwedge\{P_i \in \mathfrak{M} \colon \exists a \in A(n) \colon a.\mathit{elem} = P_i\} - \bigvee\{P_i \in \mathfrak{M} \colon \forall a' \in A(n) \colon a'.\mathit{elem} \not= P_i\}$, which is of the above form $\phi$ after distributing negation through disjunction.

To prove our claim, first note that $\mathsf{PEC}(n) \equiv n.\mathit{elem} - \bigvee_{d \in D(n)}{d.\mathit{elem}}$ where $D(n)$ is the strict subtree of $n$, i.e., $D(n)$ contains all nodes $n'$ such that $n'.\mathit{elem} \subset n.\mathit{elem}$.

Second, note that for any node $m$, $\forall a \in A(m) \colon m.\mathit{elem} \subseteq a.\mathit{elem}$, since the edges in the DAG represent subset-inclusion. So $n.\mathit{elem} \equiv \bigwedge_{a \in A(n)}{a.\mathit{elem}}$. But note that for any ancestor $a$ in $A(n)$ such that $a.\mathit{elem} \not\in \mathfrak{M}$ (i.e., $a$ is created as a result of closure under intersection), we can replace $a$ with the conjunction of the elements of $A(a)$, which are still in $A(n)$. We can repeat this (finite) process until we get a conjunction that only comprises DAG nodes $a \in A(n)$ such that $a.\mathit{elem} \in \mathfrak{M}$, i.e., we can express $n.\mathit{elem}$ through a conjunction of DAG nodes whose elements are in the set of input predicates. Thus $n.\mathit{elem} \equiv \bigwedge\{P_i \in \mathfrak{M} \colon \exists a \in A(n) \colon a.\mathit{elem} = P_i\}$.

Now let $n' \not\in (A(n) \cup D(n))$ be arbitrary. In other words, $n'$ is a DAG node that is neither an ancestor nor descendant of $n$ and, by definition of $A(n)$, $n' \not= n$. Note that $\exists d_{n'} \in D(n) \colon n.\mathit{elem} \land n'.\mathit{elem} \equiv d_{n'}.\mathit{elem}$, because \sharppec's lattice is closed under intersection. Hence, $n.\mathit{elem} - d_{n'}.\mathit{elem} \equiv n.\mathit{elem} - (d_{n'}.\mathit{elem} \lor n'.\mathit{elem})$. Thus, $n.\mathit{elem} - \bigvee_{d \in D(n)}{d.\mathit{elem}} \equiv n.\mathit{elem} - (\bigvee_{d \in D(n)}{d.\mathit{elem}} \lor \bigvee_{n' \not\in (A(n) \cup D(n))}{n'.\mathit{elem}}) \equiv n.\mathit{elem} - \bigvee_{a'\not\in A(n)}{a'.\mathit{elem}}$. 
Note that for any $a' \not\in A(n)$ there exists $a'' \in A(a') - A(n)$ such that $a''.\mathit{elem} \in \mathfrak{M}$: consider $b \in A(a') - A(n)$ such that all of its immediate parents are in $A(n)$. Note that such node exists because the root of the DAG is in $A(n)$. It must be the case that $b.\mathit{elem} \in \mathfrak{M}$, otherwise $b.\mathit{elem} \equiv p_1.\mathit{elem} \land p_2.\mathit{elem}$ where $p_1,p_2 \in A(n)$ which in turn means $b \in A(n)$ (a contradiction). So we can set $a'' = b$. Note that $a'.\mathit{elem} \lor a''.\mathit{elem} \equiv a''.\mathit{elem}$. By application of this observation to all $a' \not\in A(n)$ we get $\bigvee_{a'\not\in A(n)}{a'.\mathit{elem}} \equiv \bigvee\{P_i \in \mathfrak{M} \colon \exists a' \not\in A(n) \colon a'.\mathit{elem} = P_i\}  \equiv \bigvee\{P_i \in \mathfrak{M} \colon \forall a' \in A(n) \colon a'.\mathit{elem} \not= P_i\}$.    

Putting all together, $\mathsf{PEC}(n)$ can be re-written as $n.\mathit{elem} - \bigvee_{a'\not\in A(n)}{a'.\mathit{elem}} \equiv  \bigwedge\{P_i \in \mathfrak{M} \colon \exists a \in A(n) \colon a.\mathit{elem} = P_i\} - \bigvee\{P_i \in \mathfrak{M} \colon \forall a' \in A(n) \colon a'.\mathit{elem} \not= P_i\}$.
\end{proof}

\section{Detailed Experimental Results}
\label{appendix:evaluation}

This section details our experimental results, including run-time and memory usage of the different \gls*{pec}-emptiness solutions. In addition, our experiments compare \sharppec\ to both APV and \ddnf.

\begin{figure*}[t]
\centering
\scalebox{0.65}{
\begin{tabular}{@{}lllllllllllllll@{}}
\toprule
Dataset            & Insertions & PECs      & \begin{tabular}[c]{@{}l@{}}Empty\\ PECs\end{tabular} & \begin{tabular}[c]{@{}l@{}}Atomic\\ Preds.\end{tabular} & \multicolumn{2}{l|}{\begin{tabular}[c]{@{}l@{}}PEC-construction \\ time (s)\end{tabular}} & \multicolumn{3}{l}{PEC-emptiness check (s)} & APV (s) & \multicolumn{4}{c}{Memory (MB)} \\ \cmidrule(lr){6-10} \cmidrule(l){12-15}
                   &            &           &                                                      &                                                         & Z3 ddNF                           & \multicolumn{1}{l|}{\#PEC}                    & BDD       & SAT       & Card.              &         & BDD   & SAT   & Card. & APV     \\ \midrule
REANNZ-IP          & 1,159      & 1,160     & 25                                                   & 1,135                                                   & \textless{}1ms                    & \textless{}1ms                                & 0.016     & 0.414     & \textless{}1ms     & 0.001   & 6     & 6     & 3     & 5       \\
REANNZ-Full        & 1,170      & 12,783    & 275                                                  & 12,508                                                  & 0.112                             & 0.009                                         & 2         & 9         & 0.018              & 3       & 14    & 26    & 9     & 10      \\
Azure-DC           & 2,942      & 5,096,869 & 10,450                                               & 5,086,419                                               & 3301                              & 121                                           & 20112     & 47829     & 30                 & 25669   & 4,429 & 5,797 & 2,365 & 2,517   \\
Berkeley-IP        & 584,944    & 584,945   & 29,813                                               & Timeout                                                 & Timeout                           & 2709                                          & 1553      & 460       & 0.515              & Timeout & 302   & 701   & 227   & Timeout \\
Stanford-IP/bbra   & 1,825      & 1,825     & 10                                                   & 1,815                                                   & \textless{}1ms                    & \textless{}1ms                                & 0.019     & 0.657     & \textless{}1ms     & 0.041   & 6     & 7     & 3     & 5       \\
Stanford-IP/bbrb   & 1,620      & 1,620     & 8                                                    & 1,612                                                   & \textless{}1ms                    & \textless{}1ms                                & 0.017     & 0.566     & \textless{}1ms     & 0.033   & 6     & 7     & 3     & 5       \\
Stanford-IP/boza   & 1,614      & 1,614     & 3                                                    & 1,611                                                   & \textless{}1ms                    & \textless{}1ms                                & 0.018     & 0.582     & \textless{}1ms     & 0.039   & 6     & 7     & 3     & 5       \\
Stanford-IP/bozb   & 1,453      & 1,453     & 2                                                    & 1,451                                                   & \textless{}1ms                    & \textless{}1ms                                & 0.017     & 0.521     & \textless{}1ms     & 0.033   & 6     & 6     & 3     & 5       \\
Stanford-IP/coza   & 184,909    & 184,909   & 4,840                                                & 180,069                                                 & 471                               & 334                                           & 7         & 82        & 0.122              & 4911    & 102   & 252   & 69    & 49      \\
Stanford-IP/cozb   & 183,376    & 183,376   & 4,840                                                & 178,536                                                 & 465                               & 327                                           & 15        & 83        & 0.121              & 4924    & 100   & 252   & 68    & 49      \\
Stanford-IP/goza   & 1,767      & 1,767     & 1                                                    & 1,766                                                   & \textless{}1ms                    & \textless{}1ms                                & 0.021     & 0.639     & \textless{}1ms     & 0.045   & 6     & 7     & 3     & 5       \\
Stanford-IP/gozb   & 1,669      & 1,669     & 1                                                    & 1,668                                                   & \textless{}1ms                    & \textless{}1ms                                & 0.02      & 0.603     & \textless{}1ms     & 0.041   & 6     & 7     & 3     & 5       \\
Stanford-IP/poza   & 1,489      & 1,489     & 1                                                    & 1,488                                                   & \textless{}1ms                    & \textless{}1ms                                & 0.017     & 0.532     & \textless{}1ms     & 0.033   & 6     & 6     & 3     & 5       \\
Stanford-IP/pozb   & 1,434      & 1,434     & 1                                                    & 1,433                                                   & \textless{}1ms                    & \textless{}1ms                                & 0.017     & 0.514     & \textless{}1ms     & 0.032   & 6     & 6     & 3     & 5       \\
Stanford-IP/roza   & 1,567      & 1,567     & 2                                                    & 1,565                                                   & \textless{}1ms                    & \textless{}1ms                                & 0.018     & 0.57      & \textless{}1ms     & 0.039   & 6     & 7     & 3     & 5       \\
Stanford-IP/rozb   & 1,483      & 1,483     & 1                                                    & 1,482                                                   & \textless{}1ms                    & \textless{}1ms                                & 0.017     & 0.531     & \textless{}1ms     & 0.034   & 6     & 6     & 3     & 5       \\
Stanford-IP/soza   & 184,682    & 184,682   & 4,841                                                & 179,841                                                 & 471                               & 347                                           & 7         & 82        & 0.119              & 4951    & 102   & 251   & 69    & 49      \\
Stanford-IP/sozb   & 180,944    & 180,944   & 4,841                                                & 176,103                                                 & 443                               & 315                                           & 9         & 83        & 0.12               & 4711    & 99    & 250   & 68    & 48      \\
Stanford-IP/yoza   & 4,746      & 4,746     & 3                                                    & 4,743                                                   & \textless{}1ms                    & \textless{}1ms                                & 0.076     & 2         & 0.002              & 2       & 8     & 9     & 4     & 6       \\
Stanford-IP/yozb   & 2,592      & 2,592     & 1                                                    & 2,591                                                   & \textless{}1ms                    & \textless{}1ms                                & 0.036     & 0.969     & 0.001              & 0.303   & 6     & 7     & 3     & 5       \\
Stanford-IP/All    & 197,828    & 197,828   & 4,874                                                & 192,954                                                 & 266                               & 199                                           & 19        & 89        & 0.156              & 5149    & 122   & 265   & 85    & 53      \\
Stanford-Full/bbra & 918        & 43,450    & 0                                                    & 43,450                                                  & 2                                 & 0.361                                         & 2         & 20        & 0.049              & 12      & 39    & 38    & 20    & 26      \\
Stanford-Full/bbrb & 861        & 16,017    & 0                                                    & 16,017                                                  & 0.221                             & 0.001                                         & 0.552     & 7         & 0.019              & 3       & 17    & 17    & 9     & 12      \\
Stanford-Full/boza & 316        & 23,230    & 0                                                    & 23,230                                                  & 0.076                             & \textless{}1ms                                & 0.91      & 11        & 0.018              & 3       & 23    & 23    & 12    & 17      \\
Stanford-Full/bozb & 286        & 19,662    & 0                                                    & 19,662                                                  & 0.043                             & \textless{}1ms                                & 0.719     & 9         & 0.014              & 2       & 20    & 20    & 10    & 15      \\
Stanford-Full/coza & 417        & 14,120    & 0                                                    & 14,120                                                  & 0.032                             & \textless{}1ms                                & 0.547     & 7         & 0.02               & 2       & 16    & 18    & 9     & 12      \\
Stanford-Full/cozb & 346        & 9,200     & 0                                                    & 9,200                                                   & 0.029                             & \textless{}1ms                                & 0.346     & 4         & 0.011              & 1       & 11    & 13    & 6     & 9       \\
Stanford-Full/goza & 326        & 26,396    & 0                                                    & 26,396                                                  & 0.149                             & 0.016                                         & 1         & 12        & 0.021              & 3       & 26    & 26    & 14    & 19      \\
Stanford-Full/gozb & 306        & 23,202    & 0                                                    & 23,202                                                  & 0.11                              & 0.014                                         & 1         & 10        & 0.017              & 2       & 23    & 22    & 11    & 17      \\
Stanford-Full/poza & 243        & 14,883    & 0                                                    & 14,883                                                  & 0.017                             & \textless{}1ms                                & 0.532     & 7         & 0.011              & 1       & 16    & 16    & 9     & 12      \\
Stanford-Full/pozb & 230        & 13,284    & 0                                                    & 13,284                                                  & 0.002                             & \textless{}1ms                                & 0.473     & 6         & 0.009              & 1       & 15    & 15    & 8     & 11      \\
Stanford-Full/roza & 181        & 7,933     & 0                                                    & 7,933                                                   & \textless{}1ms                    & \textless{}1ms                                & 0.254     & 3         & 0.006              & 0.567   & 11    & 11    & 5     & 8       \\
Stanford-Full/rozb & 166        & 6,930     & 0                                                    & 6,930                                                   & 0.001                             & \textless{}1ms                                & 0.216     & 3         & 0.005              & 0.421   & 10    & 10    & 5     & 8       \\
Stanford-Full/soza & 524        & 16,764    & 81                                                   & 16,683                                                  & 0.056                             & \textless{}1ms                                & 0.668     & 9         & 0.024              & 2       & 18    & 19    & 10    & 13      \\
Stanford-Full/sozb & 355        & 9,238     & 64                                                   & 9,174                                                   & 0.028                             & \textless{}1ms                                & 0.333     & 4         & 0.011              & 0.828   & 11    & 12    & 6     & 9       \\
Stanford-Full/yoza & 507        & 60,363    & 231                                                  & 60,132                                                  & 5                                 & 0.17                                          & 4         & 38        & 0.17               & 20      & 46    & 65    & 31    & 28      \\
Stanford-Full/yozb & 353        & 27,313    & 208                                                  & 27,105                                                  & 0.97                              & 0.001                                         & 2         & 16        & 0.066              & 6       & 23    & 33    & 15    & 14      \\
Stanford-Full/All  & 2,732      & 1,176,095 & 48,906                                               & 1,127,189                                               & 560                               & 28                                            & 692       & 1958      & 4                  & 2314    & 895   & 1,077 & 544   & 439     \\ \bottomrule
\end{tabular}
}
\caption{Evaluation results for datasets that encode match conditions as ternary bit vectors}\label{fig:evaluation-results-appendix}
\end{figure*}

\begin{figure*}[]
\centering
\scalebox{0.65}{
\begin{tabular}{@{}llllllllllllll@{}}
\toprule
Dataset    & Insertions & PECs      & \begin{tabular}[c]{@{}l@{}}Empty\\ PECs\end{tabular} & \begin{tabular}[c]{@{}l@{}}Atomic \\ Preds.\end{tabular} & \begin{tabular}[c]{@{}l@{}}PEC-construction\\ time (s)\end{tabular} & \multicolumn{3}{l}{PEC-emptiness check (s)} & APV (s) & \multicolumn{4}{c}{Memory (MB)} \\ \cmidrule(lr){7-9} \cmidrule(l){11-14}
           &            &           &                                                      &                                                          &                                                              & BDD       & SAT       & Card.              &         & BDD    & SAT    & Card.  & APV  \\ \midrule
Diekmann/A & 45         & 66        & 6                                                    & 60                                                       & 0.003                                                        & 0.006     & 0.418     & \textless{}1ms     & 0.004   & 5      & 11     & 2      & 4    \\
Diekmann/B & 51         & 58        & 3                                                    & 55                                                       & 0.002                                                        & 0.005     & 0.215     & \textless{}1ms     & 0.003   & 5      & 12     & 2      & 4    \\
Diekmann/C & 31         & 92        & 0                                                    & 92                                                       & 0.001                                                        & 0.005     & 0.333     & \textless{}1ms     & 0.008   & 5      & 12     & 1      & 4    \\
Diekmann/D & 262        & 3,630     & 0                                                    & 3,630                                                    & 0.048                                                        & 0.568     & 14        & 0.014              & 1       & 14     & 23     & 11     & 8    \\
Diekmann/E & 98         & 344       & 0                                                    & 344                                                      & 0.004                                                        & 0.021     & 1         & 0.002              & 0.039   & 6      & 13     & 2      & 4    \\
Diekmann/F & 5,317      & 888,652   & 40                                                   & 888,612                                                  & 64                                                           & 408       & 4799      & 11                 & 2988    & 7,436  & 7,448  & 7,517  & 608  \\
Diekmann/G & 5,321      & 889,646   & 40                                                   & 889,606                                                  & 39                                                           & 413       & 4729      & 10                 & 2385    & 3,843  & 3,854  & 3,924  & 608  \\
Diekmann/H & 5,463      & 919,353   & 56                                                   & 919,297                                                  & 40                                                           & 424       & 5005      & 11                 & 2547    & 3,911  & 3,923  & 3,995  & 628  \\
Diekmann/I & 5,426      & 908,849   & 56                                                   & 908,793                                                  & 40                                                           & 417       & 4799      & 11                 & 2478    & 3,865  & 3,877  & 3,948  & 622  \\
Diekmann/J & 6,004      & 1,058,897 & 56                                                   & 1,058,841                                                & 71                                                           & 486       & 5654      & 13                 & 2936    & 4,558  & 4,573  & 4,656  & 700  \\
Diekmann/K & 3,242      & 400,911   & 257                                                  & 400,654                                                  & 18                                                           & 157       & 2084      & 3                  & 732     & 1,997  & 2,006  & 2,031  & 233  \\
Diekmann/L & 3,724      & 433,399   & 198                                                  & 433,201                                                  & 19                                                           & 174       & 2284      & 3                  & 921     & 2,200  & 2,209  & 2,236  & 262  \\
Diekmann/M & 136        & 426       & 0                                                    & 426                                                      & 0.009                                                        & 0.021     & 1         & 0.003              & 0.028   & 5      & 12     & 2      & 4    \\
Diekmann/N & 136        & 418       & 0                                                    & 418                                                      & 0.006                                                        & 0.021     & 1         & 0.002              & 0.027   & 5      & 12     & 2      & 4    \\
Diekmann/O & 569        & 314,160   & 0                                                    & 314,160                                                  & 30                                                           & 100       & 1149      & 2                  & 345     & 912    & 921    & 937    & 245  \\
Diekmann/P & 578        & 492,378   & 4                                                    & 492,374                                                  & 47                                                           & 168       & 1837      & 4                  & 635     & 1,563  & 1,573  & 1,606  & 324  \\
Diekmann/Q & 307        & 4,626     & 38                                                   & 4,588                                                    & 0.087                                                        & 0.763     & 17        & 0.016              & 0.94    & 21     & 29     & 18     & 7    \\
Diekmann/R & 36         & 85        & 0                                                    & 85                                                       & \textless{}1ms                                               & 0.006     & 0.311     & \textless{}1ms     & 0.013   & 5      & 12     & 1      & 4    \\
Diekmann/S & 332        & 792       & 0                                                    & 792                                                      & 0.014                                                        & 0.023     & 3         & 0.004              & 0.167   & 7      & 14     & 4      & 4    \\
Diekmann/T & 2,343      & 8,878     & 0                                                    & 8,878                                                    & 0.132                                                        & 0.341     & 33        & 0.076              & 11      & 48     & 55     & 46     & 8    \\
Diekmann/U & 73         & 93        & 0                                                    & 93                                                       & 0.002                                                        & 0.004     & 0.334     & \textless{}1ms     & 0.002   & 5      & 12     & 2      & 5    \\
Diekmann/V & 43         & 65        & 0                                                    & 65                                                       & 0.001                                                        & 0.005     & 0.236     & \textless{}1ms     & 0.007   & 5      & 12     & 2      & 4    \\
Diekmann/W & 35         & 34        & 0                                                    & 34                                                       & 0.001                                                        & 0.003     & 0.125     & \textless{}1ms     & 0.002   & 5      & 11     & 1      & 4    \\
Diekmann/X & 92         & 78        & 0                                                    & 78                                                       & 0.002                                                        & 0.006     & 0.265     & \textless{}1ms     & 0.008   & 5      & 12     & 2      & 4    \\ \bottomrule
\end{tabular}
}
\caption{Evaluation results for iptables rule-sets}\label{fig:diekmann-evaluation-results-appendix}
\end{figure*}

\end{appendices}

\end{document}